\documentclass[12pt]{iopart}

\usepackage{iopams}  

\usepackage[english]{babel}

\usepackage{mathrsfs}

\expandafter\let\csname equation*\endcsname\relax
\expandafter\let\csname endequation*\endcsname\relax
\usepackage{amsmath}

\usepackage{mathabx}

\usepackage{amsthm}

\usepackage{amsfonts}

\usepackage{BOONDOX-cal}

\usepackage{graphicx}

\newtheorem{theorem}{Theorem}[section] 
\newtheorem{lemma}[theorem]{Lemma}     

\newcommand{\spacetimeMfd}{\mathcal{M}}

\newcommand{\spatialSubMfd}{\Sigma}

\newcommand{\graphSymmetries}{\mathrm{GS}}

\newcommand{\diff}{\mathrm{Diff}}

\newcommand{\tdiff}{\mathrm{TDiff}}

\newcommand{\sym}{\mathrm{S}}

\newcommand{\diffState}[1]{\left|#1\right)}

\newcommand{\dualDiffState}[1]{\left(#1\right|}

\newcommand{\ket}[1]{\left|#1\right>}

\newcommand{\innerproduct}[2]{\left<\,#1 \, , \, #2 \,\right>}

\newcommand{\norm}[1]{||#1||}

\newcommand{\fockstate}[1]{\diffState{\,\mathfrak{#1}\,}}

\newcommand{\fockSpaceElement}[1]{\mathfrak{#1}}

\newcommand{\fockProduct}[2]{\left(\,\mathfrak{#1}\, ,  \,\mathfrak{#2}\,\right)}

\newcommand{\fockOperator}[1]{\,\mathfrak{#1}}

\newcommand{\dualFockState}[1]{\left(\,\mathfrak{#1}\,\right|}

\newcommand{\fockSpace}{\mathfrak{F}}

\newcommand{\hilbertSpaceSub}[1]{\mathscr{H}_{#1}'}

\newcommand{\hDiff}{\mathscr{H}_{\mathrm{diff}}}

\newcommand{\normsquared}[1]{\norm{#1}^{2}}

\newcommand{\expectation}[1]{\left< #1 \right>}

\newcommand{\fockIsoMap}{\mathfrak{T}}

\newcommand{\creationOperator}{\fockOperator{a}^{\dagger}}

\newcommand{\annihilationOperator}{\fockOperator{a}}

\newcommand{\fiducialState}{\fockstate{0}}

\newcommand{\snf}{\mathcal{S}}

\newcommand{\Comp}[1]{\mathrm{Comp}\left(#1\right)}

\newcommand{\curlybrackets}[1]{\left\lbrace #1 \right\rbrace}

\counterwithin{equation}{section}

\newcommand{\quotes}[1]{``#1''}

\newcommand{\startNewParagraph}{\newline\newline\noindent}

\newtheorem{assumption}{Assumption}
\newtheorem{corollaryForDiffeoInnerProduct}{Corollary}

\begin{document}

\title[]{A Fock space structure for the diffeomorphism invariant Hilbert space of loop quantum gravity and its applications}

\author{Hanno Sahlmann $^1$, Waleed Sherif $^2$}

\address{Institute for Quantum Gravity, Department of Physics, Friedrich-Alexander-Universit\"{a}t Erlangen-N\"{u}rnberg (FAU), Staudtstraße 7, 91058 Erlangen, Germany}
\ead{$^1$ hanno.sahlmann@gravity.fau.de, $^2$ waleed.sherif@gravity.fau.de}
\vspace{10pt}
\begin{indented}
\item[]February 7, 2023
\end{indented}

\begin{abstract}
Loop quantum gravity (LQG) is a quantization program for gravity based on the principles of QFT and general covariance of general relativity. Quantum states of LQG describe gravitational excitations based on graphs embedded in a spatial slice of spacetime. We show that, under certain assumptions on the class of diffeomorphisms, the space of diffeomorphism invariant states carries a Fock space structure. The role of one-particle excitations for this structure is played by the diffeomorphism invariant states based on graphs with a single (linked) component. This means, however, that a lot of the structure of the diffeomorphism invariant Hilbert space remains unresolved by this structure. We show how the Fock structure allows to write at least some condensate states of group field theory as diffeomorphism invariant coherent states of LQG in a precise sense. We also show how to construct other interesting states using this Fock structure. We finally explore the quantum geometry of single- and multi-particle states and tentatively observe some resemblance to geometries with a single or multiple components, respectively. 
\end{abstract}

%
%
%
%
%

\section{Introduction}
A fully successful and consistent quantum theory of gravity is still elusive. 
Several theories have been proposed \cite{Oriti:2009zz}, each with varying degrees of results as well as viewpoints and technical issues. Loop quantum gravity (LQG) is a pragmatic quantization program for gravity based on the principles of QFT and general covariance of general relativity. 
It offers an array of substantial results that stand on a solid mathematical foundation \cite{Ashtekar:2004eh,Thiemann:2007pyv,Rovelli:2011eq}.
\startNewParagraph
The quantum states of LQG are unusual, describing gravitational excitations based on graphs embedded in a spatial slice of spacetime. Very roughly speaking, meeting points (vertices) of these excitations create spatial volume while the excitations themselves are flux tubes of spatial area. In group field theory (GFT) a formalism has been developed in which vertex states are the single particle states of a multi-particle theory \cite{Freidel:2005qe,Oriti:2013aqa}, with many particle condensates describing, for example, cosmology \cite{Gielen:2013kla,Gielen:2013naa,Marchetti:2020umh,Jercher:2021bie}. The states of GFT are organized into a Fock space, and are very close to those of LQG, but there are some slight technical differences, mainly regarding the embedding of the networks \cite{Oriti:2013aqa,Oriti:2014yla,Oriti:2014uga,Oriti:2017ave}. 
\startNewParagraph
Connections to a multi-particle picture and Fock spaces has also been made along other lines. It has been shown that the Fock spaces for free field theory can be cast in the mathematical framework of LQG \cite{Varadarajan:1999it,Varadarajan:2001nm,Varadarajan:2002ht,Thiemann:2002vj}, and more recently \cite{Assanioussi:2022rkf}. 
\startNewParagraph
Moreover, certain subspaces of the 
the Hilbert space of \emph{partial solutions} to the vector constraint $\mathscr{H}^{\mathrm{vtx}}$ \cite{Lewandowski:2014hza}, can be given the structure of Fock spaces \cite{Assanioussi:2018zit,Assanioussi:2020fsz}. More precisely, 
it is shown that one can decompose $\mathscr{H}^{\mathrm{vtx}}$ into separable subspaces, motivated by the action of the Hamiltonian constraint with a specifically chosen regularization. These subspaces are defined over \textit{ancestor graphs}, graphs which satisfy certain criteria and from which the remaining graphs can be generated (see \cite{Assanioussi:2018zit,Assanioussi:2020fsz} for details). In the language of multi-particle systems, the \quotes{particles} created in this case would be local excitations at the vertices of these graphs and the multi-particle system is one which is a direct sum of Fock spaces over several ancestor graphs.
\startNewParagraph
Finally, it has also been shown that when quantizing a scalar field with the methods of LQG, the diffeomorphism invariant Hilbert space $\hDiff$ takes the form of a Fock space \cite{Sahlmann:2006qs}. 
\startNewParagraph
In the present work, we contribute to the multi-particle picture for LQG. 
We show that under certain assumptions on the category of manifolds used, the space of diffeomorphism invariant states $\hDiff$
carries a Fock space structure. The role of one-particle excitations for this structure is played by the diffeomorphism invariant states based on graphs consisting of a \emph{single component}\footnote{We should caution the reader that by single component graphs, we mean both, graphs $\gamma$ with trivial  zeroth homotopy group $\pi_0(\gamma)$ and graphs with several topological components that are linked, however.}. States based on multiple unlinked components furnish the multi-particle states. The result is mathematically rigorous, it uses the precise definition of $\hDiff$ \cite{Ashtekar:1995zh,Ashtekar:2004eh}. Therefore, the precise nature of the category of manifolds used is relevant for the result. Since several categories (for example: analytic \cite{Ashtekar:1993wf,Ashtekar:1995zh}, semianalytic \cite{Lewandowski:2005jk}, stratified \cite{Fleischhack:2004jc}, smooth \cite{Baez:1995zx,Lewandowski:1999qr,Fleischhack:2003vk}) have been used in LQG, we have chosen to base the analysis on a number of assumptions about the ability of diffeomorphisms to realize certain maps between graphs in a local fashion. We have no doubt that they are true at least in the semianalytic category, but we have given no formal proof. 
\startNewParagraph
We should emphasize however, that in basing the Fock structure on single component graphs, it necessarily leaves a lot of the structure of the the diffeomorphism invariant Hilbert space unresolved. The structure of a single graph component can be arbitrarily complex, and we do nothing to reduce this complexity. Inside a single component, methods such as \cite{Assanioussi:2018zit,Assanioussi:2020fsz} should be used to understand the structure further.
\startNewParagraph
We show that our result allows to write at least some of the interesting multi-particle states of GFT as diffeomorphism invariant states of LQG in a precise sense. We also demonstrate how to construct other interesting states using the new Fock structure. The group averaging procedure used in the definition of $\hDiff$ is quite technical, and thus a result that gives a more direct picture of aspects of the Hilbert space structure is useful. We also explore the quantum geometry of single- and multi-particle states, by considering generic foliations of the spatial manifold and their quantum geometric properties. By comparing these with the geometric properties of such foliations in classical Riemannian manifolds, we collect some evidence that single component graphs correspond to single component Riemannian manifolds, whereas  states based on  graphs with multiple components not linked to each other correspond to Riemannian manifolds with multiple components. 
\startNewParagraph
In more detail, the content of the paper is as follows:
In section \ref{section_isomorphismProof}, the isomorphism between the Fock space $\fockSpace$ and $\hDiff$ is explicitly derived, preceded by the introduction of some necessary mathematical tools. We begin pragmatically by first setting the mathematical stage in terms of defining the commonly used notation, terminology and spaces. Next, we decompose the group of graph symmetries $\graphSymmetries$ over an arbitrary graph and study the symmetries of the arising subgroup structures and the decomposition of its Haar measure. Utilizing these results, the diffeomorphism invariant inner product on $\hDiff$ is decomposed into its component form (that is, a decomposition into a direct sum of Hilbert spaces over the number of mutually unlinked components of the graph). Subsequently, we construct the component Fock space $\fockSpace$ and an explicit isomorphism between $\hDiff$ and $\fockSpace$ is then established.
\startNewParagraph
Section \ref{section_applications} concerns the physical applications of the states defined on $\fockSpace$. We begin by demonstrating that one can obtain coherent states with interesting properties which would rather be difficult or not possible to obtain using the standard mathematical language of LQG. Following that, we relate the results obtained with the work conducted in the group field theory approach to quantum gravity \cite{Freidel:2005qe,Oriti:2013aqa}, specifically, in relation to the work on condensates \cite{Gielen:2013kla,Gielen:2013naa,Marchetti:2020umh,Jercher:2021bie}. Lastly, we conclude by providing a geometric interpretation of the general states defined on $\fockSpace$. This is done by presenting two case studies in which we qualitatively look at two candidates for two geometric observables, the area and volume operators, where the geometry one obtains using this Fock structure can be explored.

\section{A Fock structure for the diffeomorphism invariant Hilbert space of loop quantum gravity}
\label{section_isomorphismProof}
In LQG, the category of smoothness for the spacetime and its diffeomorphisms plays a role in some technical developments. Several proposals for this category have been made, for example: analytic \cite{Ashtekar:1993wf,Ashtekar:1995zh}, semianalytic \cite{Lewandowski:2005jk}, stratified \cite{Fleischhack:2004jc}, smooth \cite{Baez:1995zx,Lewandowski:1999qr,Fleischhack:2003vk}. Rather than working with a particular one, we will make some assumption about the class of diffeomorphisms in the following (see Assumption \ref{AssumptionsForTheFockSpaceIsomorphism} below). States of the gravitational field contain one-dimensional excitations contained in a spatial hypersurface $\spatialSubMfd$. To keep things non-technical we assume that topologically $\spatialSubMfd\simeq \mathbb{R}^3$. 
\startNewParagraph
We introduce some important notions and notations. In the following, we denote by $\Gamma$ the set of all graphs. Furthermore, for $\gamma \in \Gamma$, we denote by $\Comp{\gamma}$ the set of \emph{unlinked components} of $\gamma$. Here, by an unlinked component of $\gamma$, we mean one or several components of $\gamma$ that are not linked to the rest of the graph, but are linked with each other. We call two parts of the graph unlinked if we can find diffeomorphsims that fix the one part while moving the other arbitrarily far away in terms of some fiducial flat metric on $\mathbb{R}^3$. 
The set of all graphs composed of $n$ unlinked components is denoted
\begin{equation}
   \Gamma_{n} := \left\lbrace \gamma \in \Gamma | |\Comp{\gamma}| = n \in \mathbb{N}_{0} \right\rbrace.
\end{equation}
We now introduce the notation necessary for the definition of $\hDiff$ following \cite{Ashtekar:2004eh}. 
Denote by $\mathrm{Diff}_{\gamma} \subseteq \mathrm{Diff}_{\spatialSubMfd}$ the group of diffeomorphisms which map the graph $\gamma$ onto itself. Furthermore, denote by $\mathrm{TDiff}_{\gamma}$ the group of trivial diffeomorphisms of $\gamma$ which maps it onto itself while preserving each edge $e \in E(\gamma)$ of the graph and its orientation. Lastly,  $\graphSymmetries_{\gamma} = \mathrm{Diff}_{\gamma}/\mathrm{TDiff}_{\gamma}$ denotes the group of graph symmetries of $\gamma$ \cite{Ashtekar:2004eh}. 
\startNewParagraph
States of LQG are defined in terms of holonomies $h$ of an $\mathrm{SU}(2)$ connection $A$ along edges of graphs. Functionals $f[A] \in \mathrm{Cyl_{\gamma}^{\infty}}$ are called smooth cylindrical functions with respect to $\gamma$ if there exists a smooth function $f:\mathrm{SU}(2)^{|E(\gamma)|} \rightarrow \mathbb{C}$ such that $f[A] = f(h_{e_1}[A], \cdots , h_{e_{|E(\gamma)|}}[A])$ \cite{Ashtekar:1993wf}. The space $\mathrm{Cyl}_{\gamma}^{\infty}$ has a natural inner product with respect to a Haar measure $d\mu_{H}^{0}$ \cite{Ashtekar:1994mh}.
Its Cauchy completion with respect to the Ashtekar-Lewandowski measure gives a Hilbert space, here denoted by $\hilbertSpaceSub{\gamma}$, which is isomorphic to the uniform measure space $L^{2}(\Bar{\mathcal{A}}, d\mu_{\gamma}^{0})$, where $\Bar{\mathcal{A}}$ is the compact Hausdorff space of generalized connections \cite{Ashtekar:1994mh}. Therefore, one constructs a Hilbert space $\mathscr{H}$ over all graphs such that
\begin{equation}
    \mathscr{H} = \bigoplus_{\gamma \in \Gamma} \hilbertSpaceSub{\gamma}.
\end{equation}
Here, the space $\hilbertSpaceSub{\gamma}$ is a space spanned by gauge  invariant spin networks over $\gamma$ which fulfill some additional conditions  (no trivial two-valent vertices, no trivial representations) \cite{Ashtekar:2004eh}. 
Since the inner product provides orthogonality between the spaces constructed on any two different graphs $\gamma \in \Gamma_{n}$, $\gamma' \in \Gamma_{m}$ where $n \neq m$, then
\begin{equation}
\label{eqn_defHdiff}
    \mathscr{H} = \bigoplus_{n = 0}^{\infty}\bigoplus_{\gamma \in \Gamma_{n}}\hilbertSpaceSub{\gamma} =: \bigoplus_{n = 0}^{\infty}\hilbertSpaceSub{n}.
\end{equation}
To obtain diffeomorphism invariant Hilbert spaces, and thus the space of solutions to the diffeomorphism constraint of LQG, one needs to group average the cylindrical functions with respect to the diffeomorphisms of the graph. This is done in two steps. First, for any given graph $\gamma$, one averages with respect to $\graphSymmetries_{\gamma}$, hence projecting elements of $\hilbertSpaceSub{\gamma}$ onto its subspace which is invariant under the action of $\graphSymmetries_{\gamma}$. This projection has a natural extension to all graphs. Next, we average with respect to all diffeomorphisms which move the graph. Since the group $\mathrm{Diff}_{\spatialSubMfd}$ is non-compact, this immediately leads to an issue of that the states are now no longer normalizable. Consequently, one works in the dual space of the cylindrical functions rather than the Hilbert space itself and, using a rigging map $\eta$, one assigns to every $\Psi_{\gamma} \in \hilbertSpaceSub{\gamma}$ an element $\dualDiffState{\eta(\Psi_{\gamma})} \in \mathrm{Cyl}^{\infty^{\star}}$ \footnote{Here, $\mathrm{Cyl}^{\infty^{\star}}$ is the dual space of $\mathrm{Cyl^{\infty}} := \bigcup_{\gamma \in \Gamma} \mathrm{Cyl}_{\gamma}^{\infty}$} and define it by its action on $\ket{\Phi_{\widetilde{\gamma}}} \in \mathrm{Cyl}^{\infty}$ such that (for details see \cite{Ashtekar:2004eh,Thiemann:2007pyv})
\begin{equation}
\label{eqn_diffeoInvariantFunctional}
    \left(\eta(\Psi_{\gamma})| \Phi_{\widetilde{\gamma}}\right> := \frac{1}{|\graphSymmetries_{\gamma}|}\sum_{\varrho\in\mathrm{Diff}/\mathrm{Diff}_{\gamma}}\sum_{\varphi\in\graphSymmetries_{\gamma}} \innerproduct{\varrho^{*}\varphi^{*}\Psi_{\gamma}}{\Phi_{\widetilde{\gamma}}}_{\mathscr{H}}.
\end{equation}
Using the group averaging map $\eta$ one can now turn its image into a Hilbert space by defining
\begin{equation}
\label{eqn_diffeoInvariantInnerProduct}
    \left(\eta(\Psi) | \eta(\Phi)\right) = \left(\eta(\Psi) | \Phi\right>.
\end{equation}
Since graphs with different components can never be diffeomorphic, the total diffeomorphism invariant Hilbert space is then
\begin{equation}
\label{difHilSpace}
    \hDiff := \eta(\mathscr{H}) = \bigoplus_{n = 0}^{\infty}\eta(\hilbertSpaceSub{n}).
\end{equation}
The aim now is to investigate the inner product (\ref{eqn_diffeoInvariantInnerProduct}) defined on $\hDiff$ and to reduce it into its simplest form. Systematically, we start with the group averaging procedure. In doing so, we will study the structure of the group of graph symmetries $\graphSymmetries_{\gamma}$ of some $\gamma \in \Gamma_{n}$. We begin by constructing a group homomorphism, denoted by $\rho$, between $\graphSymmetries_{\gamma}$ and $\sym_{n}$ where $\sym_{n}$ is the symmetric group of $n$ elements. Let $I := \left\lbrace 1, 2, \dots , n \right\rbrace \subseteq \mathbb{N}$ be an index set labeling elements of $\Comp{\gamma}$, i.e., we fix a 1-1 map between $I$ and $\Comp{\gamma}$. 
\startNewParagraph
We consider $[g] \in \graphSymmetries_{\gamma}$. Since every $g \in [g]$ is continuous, it must map a component $c_{i} \in \Comp{\gamma}$ to some $c_{j}\in\Comp{\gamma}$, up to permutation of its edges and change of their orientation, where $i, j \in I$. Therefore, we can define a map 
\begin{equation}
    \rho: \graphSymmetries_\gamma \rightarrow \sym_n, \qquad g (c_k) = c_{\rho([g])(k)}, \qquad g \in \diff_\gamma, k \in I. 
\end{equation}
One can easily see that $\rho$ is a group homomorphism, 
\begin{equation}
    \rho([g_{1} \circ g_{2}]) = \rho([g_{1}]) \circ \rho([g_{2}]). 
\end{equation}
We will now use this map to decompose the group of graph symmetries into two subgroups, one consisting of those symmetries which map the graph components onto themselves, and the other consisting of graph symmetries which permute components. One knows that
\begin{equation}
    \mathrm{H}_\gamma:=\ker\rho \equiv \left\lbrace [g] \in \graphSymmetries_\gamma \, | \, \rho ([g]) = \mathbb{I}_{\sym_n} \right\rbrace
\end{equation}
is a normal subgroup of $\graphSymmetries_\gamma$, and that the image $\rho(\graphSymmetries_\gamma)$ is a group isomorphic to 
\begin{equation}
  \widetilde{\graphSymmetries}_\gamma:=\graphSymmetries_\gamma/\mathrm{H}_\gamma.  
\end{equation}
Note that relabeling the components of $\gamma$ changes $\rho$ by a conjugation with an element of $\sym_n$ and does not change the kernel of $\rho$, and hence $\mathrm{H}_\gamma$ is independent of this labeling, and so is $\widetilde{\graphSymmetries}_\gamma$. Using this decomposition, we can decompose sums over group elements as follows. 
\begin{lemma}[Decomposition of the Haar measure]
\label{lemma_decompositionOfGraphSymmetries}
The summation over the elements of the group $\graphSymmetries_\gamma$ (i.e., its Haar measure) can be written in terms of equivalence classes of $\widetilde{\graphSymmetries}_\gamma$ and elements of $\mathrm{H}_\gamma$ as
\begin{equation}
\label{eqn_haar_measure}
        \sum_{g\in\graphSymmetries_\gamma} f(g) =\sum_{[g_i]\in\widetilde{\graphSymmetries}_\gamma} \sum_{h\in\mathrm{H}_\gamma} f(h {g}_i)
\end{equation}
where $\{{g}_i\}$ is a set of representatives, one for each equivalence class in $\widetilde{\graphSymmetries}_\gamma$, and $f$ is an arbitrary function on $\graphSymmetries_{\gamma}$. 
\end{lemma}
\begin{proof}
The group $\graphSymmetries_\gamma$ as a set decomposes into $\widetilde{\graphSymmetries}_\gamma \times \mathrm{H}_\gamma$. Conjugacy classes in $\widetilde{\graphSymmetries}_\gamma$ can be written as $\mathrm{H}_\gamma g_1 , \mathrm{H}_\gamma g_2, \dots$ where $\{{g}_i\}$ is a set of representatives, one for each equivalence class. In terms of these representatives, any $g'\in\graphSymmetries_\gamma$ can hence be written as $g' = h g_i$ for precisely one $g_i$ from the chosen list of representatives. $h$ is then uniquely defined as $h=g'/g_i$. As a result, one can write a sum over the group elements of the graph symmetries over the entire graph as
\begin{equation}
    \sum_{g\in\graphSymmetries_\gamma}  f(g) =  \sum_{[{g}_i]\in\widetilde{\graphSymmetries}_\gamma} \sum_{h\in\mathrm{H}_\gamma} f(hg_i),
\end{equation}
for any choice of the representatives $\{{g}_i\}$ of the conjugacy classes. 
\end{proof}
We do not yet have a very concrete picture of what the groups $\mathrm{H}_\gamma$ and $\widetilde{\graphSymmetries}_{\gamma}$ look like. Their structure obviously depends on $\gamma$, as well as on the groups $\diff_\gamma$ and $\tdiff_\gamma$, which in turn depend on the category of diffeomorphisms in use. To simplify things we will make some assumptions on the size on the diffeomorphism groups in question.  
\begin{assumption}
\label{AssumptionsForTheFockSpaceIsomorphism}
For any $c,c' \in \Gamma_1$, $\gamma \in \Gamma$ which are disjoint, mutually unlinked, with $c,c'$ diffeomorphic to each other:
\begin{itemize}
  \item[\textsc{A1.}] There exists a diffeomorphism $\varphi \in \tdiff_{\gamma}$ such that $\varphi$ maps the edges of $c$ onto the edges of $c'$ up to orientation, and the other way around, which we denote with a slight abuse of notation as 
  \begin{equation}
      \varphi(c) = c', \qquad \varphi(c') = c. 
  \end{equation}

  \item[\textsc{A2.}] For any $\varphi \in \diff_{c}$ we can find $\varphi' \in \diff_{\gamma \cup c}$ such that 
  \begin{equation}
      \varphi'\in \tdiff_\gamma, \qquad \varphi'\circ \varphi^{-1}\in \tdiff_c. 
  \end{equation}
\end{itemize}
\end{assumption}
It is useful to keep track of which of the components of $\gamma$ are diffeomorphic to each other. This partitions $\Comp{\gamma}$ as 
\begin{equation}
    \Comp{\gamma}= r_1\dot{\cup}r_2\dot{\cup}\ldots \dot{\cup} r_k
\end{equation}
where the components in $r_i$ are diffeomorphic to each other but not to any component of $\gamma$ not in $r_i$. Then we have the following.
\begin{lemma}[Symmetries of the subgroup structures of $\graphSymmetries$]
\label{lemma_structure_symmetries}
Under the assumptions \ref{AssumptionsForTheFockSpaceIsomorphism} we have 
\begin{equation}
    \mathrm{H}_\gamma \simeq \bigtimes_{c\in \Comp{\gamma} } \graphSymmetries_{c}\quad\quad \text{and} 
    \quad\quad \widetilde{\graphSymmetries}_\gamma \simeq  \bigtimes_{i=1}^k \sym_{|r_i|}.
\end{equation}
\end{lemma}
\begin{proof}
    Start with $g=[\varphi_1]\times\ldots\times[\varphi_k]\in \bigtimes_{i \in I} \graphSymmetries_{c_i}$. Using assumption A2, to each $\varphi_i$ there is $\varphi'_i\in \tdiff_{\gamma-c_i}$ with
    \begin{equation}
        \varphi'_i\circ \varphi_i^{-1}\in \tdiff_{c_i}.
    \end{equation}
    The latter means $[\varphi_i]=[\varphi'_i]$, and so $g=[\varphi'_1]\times\ldots\times[\varphi'_k]$. 
    The first isomorphism of the lemma can then be defined by 
    \begin{equation}
        g\longmapsto [\varphi'_1\circ\ldots\circ\varphi'_k]. 
    \end{equation}
    Each $\varphi'_i$ is only defined up to elements of $\tdiff_\gamma$, but those do not change the equivalence class on the right hand side and hence the map is well defined. It is also a homomorphism: For $\varphi'_i \in \tdiff_{\gamma-c_i}$, and $\xi'_j \in \tdiff_{\gamma-c_j}$ with $i\neq j$ one has $[\varphi'_i\circ \xi'_j]=[\xi'_j\varphi'_i]$ and hence 
    \begin{equation}
    \begin{split}   
        [\varphi_1]\times\ldots\times[\varphi_k] \cdot [\xi_1]\times\ldots\times[\xi_k]
        \longmapsto & [\varphi'_1\circ\xi'_1 \circ \ldots\circ\varphi'_k\circ\xi'_k]\\
        &=[\varphi'_1\circ\ldots\circ\varphi'_k\circ\xi'_1\circ\ldots\circ\xi'_k]\\
        &=[\varphi'_1\circ\ldots\circ\varphi'_k]\cdot[\xi'_1\circ\ldots\circ\xi'_k].
        \end{split}
    \end{equation}
As for the second statement, let $c_i,c_j$ be two components in $r_a$. Since they are diffeomorphic, there is, according to assumption A1 a diffeomorphism $\varphi_{ij}\in \tdiff_{\gamma -\{c_i,c_j\}}$ that interchanges $c_i$ and $c_j$. Therefore we define a map 
\begin{equation}
    \label{eqn_permutationMap}
   \widetilde{\graphSymmetries}_{\gamma}\ni [\varphi_{ij}]\longmapsto  T_{ij} \in \sym_{|r_i|}.
\end{equation}
This is well defined: If $\varphi$ is modified by an element of $\mathrm{H}_\gamma$, it will still swap $c_i$ and $c_j$. 
It can be easily checked that the map \eqref{eqn_permutationMap} respects products, so it is a homomorphism. Moreover, since A1 guarantees the existence  of a map for any two diffeomorphic components, the image of \eqref{eqn_permutationMap} has all transpositions, so it extends to a surjective map. It is also injective: If $c_i$ and $c_j$ are diffeomorphic in two inequivalent ways, one can quickly see that there is an induced graph symmetry in $\mathrm{H}_\gamma$ which can relate the two, so they actually correspond to the same element of $\widetilde{\graphSymmetries}_{\gamma}$. The arguments above can be extended from one set $r_a$ of components to all of them, thus the assertion is proven. 
\end{proof}
\noindent We have now shown that the group of graph symmetries $\graphSymmetries_\gamma$ over some graph $\gamma$ can be described in terms of the two groups shown in the lemma above. This simplifies a detailed study of the group averaging procedure. \startNewParagraph
We will now characterize the diffeomorphism invariant inner product on $\hilbertSpaceSub{n}$ in terms of the $n$ graph components. This is without loss of generality, as sectors with different $n$ are orthogonal to each other, see \eqref{eqn_defHdiff}. 
Therefore we consider two cylindrical functions $\Psi$ and $\Psi' \in \hilbertSpaceSub{\gamma}$. We can expand into products of generalized spin network functions $\snf_i^{I_i}$ over the components $c_i\in \Comp{\gamma}$:
\begin{equation}
\label{eqn_spinNetwork}
    \Psi = \sum_{I_1, \cdots , I_n}c_{I_1 , \cdots , I_n}\snf_1^{I_1} [A] \cdots \snf_n^{I_n} [A], \qquad \Psi' = \sum_{I_1 , \cdots , I_{n}}c'_{I_1 , \cdots , I_{n}} \snf_1^{I_1} [A] \cdots \snf_{n}^{I_{n}} [A]. 
\end{equation}
Here the indices $I$ represent the summation over the generalized spin networks in $\hilbertSpaceSub{c_i}$ and $n = |\Comp{\gamma}|$. For the diffeomorphism invariant inner product we find 
\begin{lemma}[Component decomposition of the inner product on $\hDiff$]
\label{lemma_resolutionOfDiffInvInnerProduct}
The inner product of $\eta(\Psi_\gamma)$ and $\eta(\Psi'_\gamma)$ on $\hDiff$ can be written as
\begin{equation}
    \left(\eta(\Psi) , \eta(\Psi')\right) = \frac{1}{|\widetilde{\graphSymmetries}_\gamma|}\sum_{\widetilde{\zeta} \in \widetilde{\graphSymmetries}_\gamma} \sum_{\left\lbrace I\right\rbrace, \left\lbrace I'\right\rbrace} \overline{c_{\rho (\widetilde{\zeta})(I_1)\dots \rho (\widetilde{\zeta})(I_n)}}c_{\left\lbrace I'\right\rbrace} \prod_{k}^{n} \left( \eta(\snf_{k}^{I_{k}}), \eta(\snf_{k}^{I_{k}'})\right)
\end{equation}
where $k\in \left\lbrace 1, \cdots , n\right\rbrace$ and $I$ are spin network labels.
\end{lemma}
\begin{proof}
By definition, the diffeomorphism invariant inner product takes the form
\begin{equation}
    \left(\eta(\Psi) , \eta(\Psi')\right) = \frac{1}{|\graphSymmetries_\gamma|}\sum_{\varphi\in\diff/\diff_\gamma}\sum_{\zeta\in\graphSymmetries_\gamma}\left<\varphi^{*}\zeta^{*}\Psi, \Psi'\right>.
\end{equation}
Following from lemma \ref{lemma_decompositionOfGraphSymmetries}, the averaging over graph symmetries can be split as in \eqref{eqn_haar_measure}, thus
\begin{align}
        \left(\eta(\Psi) , \eta(\Psi')\right) & = \begin{aligned}[t]
        \frac{1}{|\widetilde{\graphSymmetries}_{\gamma}|}\sum_{\varphi\in\diff/\mathrm{Diff_\gamma}}\sum_{\widetilde{\zeta}\in\widetilde{\graphSymmetries}_\gamma}&\left(\prod_{i}^{n}\frac{1}{|\graphSymmetries_{c_{i}}|}\right) \sum_{\zeta^{(i)}\in\mathrm{H}_{c_i}} \times
        \\
        & \times \left<\varphi^{*}{\widetilde{\zeta}}^{*}\left(\prod_{i}^{n}\zeta^{(i)*}\right)\Psi, \Psi'\right>.
        \end{aligned}
        \label{eqn_kinInnerProductTerm}
\end{align}
Since both functions in the kinematical inner product are cylindrical on the same graph, the only non-vanishing contribution is obtained for $\varphi$ in the equivalence class of the identity diffeomorphism. We further find
\begin{align}
\left<{\widetilde{\zeta}}^{*}\left(\prod_{i}^{n}\zeta^{(i)*}\right)\Psi, \Psi' \right> &= \sum_{\left\lbrace I\right\rbrace, \left\lbrace I'\right\rbrace}\overline{c_{\{I\}}} c_{\{I'\}} \left<{\widetilde{\zeta}}^{*}\left(\prod_{i}^{n}\zeta^{(i)*}\right)\prod_{k=1}^{n}\snf_k^{I_k}, \prod_{j=1}^{n}\snf_j^{I'_j}\right> \\
        &= \sum_{\left\lbrace I\right\rbrace, \left\lbrace I'\right\rbrace}\overline{c_{\left\lbrace I \right\rbrace}}c_{\left\lbrace I'\right\rbrace}\left<\widetilde{\zeta}^{*}\prod_{k=1}^{n}\zeta^{(k)*}\snf_k^{I_{k}}, \prod_{j=1}^{n}\snf_j^{I'_j}\right>
            \\
            &= \sum_{\left\lbrace I\right\rbrace, \left\lbrace I'\right\rbrace}\overline{c_{\left\lbrace I \right\rbrace}}c_{\left\lbrace I'\right\rbrace}\left<\prod_{k=1}^{n}\left(\zeta^{\rho (\widetilde{\zeta})(k)}\right)^{*}\snf^{I_k}_{\rho (\widetilde{\zeta})(k)}, \prod_{j=1}^{n}\snf^{I'_j}_j\right>
            \\
            &= \sum_{\left\lbrace I\right\rbrace, \left\lbrace I'\right\rbrace}\overline{c_{\rho (\widetilde{\zeta})(I)}}c_{\left\lbrace I'\right\rbrace}\left<\prod_{k=1}^{n}\zeta^{(k)*}{\snf}^{I_k}_{k}, \prod_{j=1}^{n}\snf_{j}^{I_{j}'}\right>
            \\
            &= \sum_{\left\lbrace I\right\rbrace, \left\lbrace I'\right\rbrace}\overline{c_{\rho (\widetilde{\zeta})(I)}}c_{\left\lbrace I'\right\rbrace}\prod_{k=1}^{n}\left<\zeta^{ (k)*}{\snf}^{I_k}_k, \snf^{I'_k}_k\right>,
\end{align}
where we have introduced the shorthand 
\begin{equation}
    \rho (\widetilde{\zeta})(I)\equiv \rho (\widetilde{\zeta})(I_1 \ldots I_n):= (I_{\rho (\widetilde{\zeta})(1)} \ldots I_{\rho (\widetilde{\zeta})(n)})
\end{equation}
for the reordering action of the permutation $\rho (\widetilde{\zeta})$ on the indices of the spin network coefficients. Therefore, we can write the inner product on $\hDiff$ as
\begin{align}
            \left(\eta(\Psi) , \eta(\Psi')\right) & =
            \begin{aligned}[t]
            \frac{1}{|\widetilde{\graphSymmetries}_\gamma|}\sum_{\varphi\in\diff/\diff_\gamma}\sum_{\widetilde{\zeta} \in \widetilde{\graphSymmetries}_\gamma} \sum_{\left\lbrace I\right\rbrace, \left\lbrace I'\right\rbrace} 
            \overline{c_{\rho (\widetilde{\zeta})(I)}}c_{\left\lbrace I'\right\rbrace} &\prod_{k=1}^{n} \frac{1}{|\graphSymmetries_{c_k}|} \times
            \\
            & \times \sum_{\zeta^{(k)}\in\mathrm{H}_{c_k}}\left<\varphi^*\zeta^{(k)*}{\snf}^{I_k}_k, \snf^{I'_k}_k\right>
            \end{aligned}
            \\
            & = \frac{1}{|\widetilde{\graphSymmetries}_\gamma|}\sum_{\widetilde{\zeta} \in \widetilde{\graphSymmetries}_\gamma} \sum_{\left\lbrace I\right\rbrace, \left\lbrace I'\right\rbrace} \overline{c_{\rho (\widetilde{\zeta})(I)}}c_{\left\lbrace I'\right\rbrace} \prod_{k=1}^{n} \left( \eta(\snf^{I_k}_k), \eta(\snf^{I'_k}_k)\right).
\end{align}
\end{proof}
\noindent We can now generalize this result slightly: For $\gamma \in \Gamma_n$ not diffeomorphic to $\gamma'\in \Gamma_n$ and $\Psi_\gamma \in \hilbertSpaceSub{\gamma}$, $\Psi'_{\gamma'} \in \hilbertSpaceSub{\gamma'}$ we have 
\begin{equation}
    \left(\eta(\Psi_\gamma) , \eta(\Psi'_{\gamma'})\right)= 0
\end{equation}
as an immediate consequence of \eqref{eqn_diffeoInvariantFunctional},\eqref{eqn_diffeoInvariantInnerProduct}. Since $\eta$ is invariant under diffeomorphisms, the case of diffeomorphic but non-equal $\gamma,\gamma'$ can be reduced to that of $\gamma=\gamma'$. Thus we can record the follwoing generalization. 
\begin{corollaryForDiffeoInnerProduct}
    For $\gamma,\gamma' \in \Gamma_n$ and $\Psi_\gamma \in \hilbertSpaceSub{\gamma}$, $\Psi'_{\gamma'} \in \hilbertSpaceSub{\gamma'}$ we have 
\begin{equation}
    \left(\eta(\Psi_\gamma) , \eta(\Psi'_{{\gamma'}})\right) = \frac{1}{|\widetilde{\graphSymmetries}_\gamma|}\sum_{\widetilde{\zeta} \in \widetilde{\graphSymmetries}_\gamma} \sum_{\left\lbrace I\right\rbrace, \left\lbrace I'\right\rbrace} \overline{c_{\rho (\widetilde{\zeta})(I)}}c_{\left\lbrace I'\right\rbrace} \prod_{k=1, j=1}^{n} \left( \eta(\snf_{k}^{I_{k}}), \eta(\widetilde{\snf}_{j}^{I_{j}'})\right).
\end{equation}
where $\widetilde{\snf}$ is a basis for $\hilbertSpaceSub{\gamma'}$ analogous to $\snf$ for $\hilbertSpaceSub{\gamma}$.
\end{corollaryForDiffeoInnerProduct}
\noindent The diffeomorphism invariant inner product thus shows a remarkable structure, as if there is an underlying tensor product over graph components, and a symmetrization over diffeomorphic components. Therefore, it is tempting to try to see if one can indeed construct an isomorphic Fock space. We therefore consider the symmetric Fock space over $\eta(\hilbertSpaceSub{1})$:
\begin{equation}
\label{eqn_theFockSpaceDef}
    \fockSpace := \mathcal{F}_{\text{S}}[\eta (\hilbertSpaceSub{1})].
\end{equation}
We call this space the \textit{component Fock space}. Let us consider $n$-particle states. We fix $n$ one-component graphs $c_k$, $k=1,2,\ldots,n$, and spin net bases $\{\snf_{k}^I\}_I$ of $\hilbertSpaceSub{c_k}$. Given the spin network bases over the one-component graphs we can expand $n$-particle states based on $c_k$, $k=1,2,\ldots,n$ as 
\begin{equation}
\label{eqn_fockStatesAnsatz}
        \fockstate{f} = \frac{1}{n!}\sum_{\sigma\in\mathrm{S}_n}\sum_{\left\lbrace I\right\rbrace} \fockSpaceElement{f}_{\left\lbrace I\right\rbrace} \bigotimes_{l=1}^{n} \left| \eta(\snf_{\gamma_{\sigma (l)}}^{I_{\sigma (l)}})\right)
\end{equation}
where $\fockSpaceElement{f}_{\{I\}} \in \mathbb{C}$ and $\snf$ are the normalized generalized spin network functions as before. 
For bookkeeping purposes we partition the components $\{c_k\}$ into sets of diffeomorphic components, with the size of the $l$\textsuperscript{th} set denoted by $m_{l}$. We have 
\begin{equation}
    \sum_{k=1}^{N} m_{k} = n,
\end{equation}
where $N$ is the total number of disjoint sets in the partition. 
We denote by $\mathrm{S}_{n}^{p}$ the subgroup of the symmetric group $\mathrm{S}_n$ that preserves the given partition. It is a subgroup, and we can write the sum over permutations in terms of its conjugacy classes, 
\begin{equation}
    \sum_{\sigma\in\mathrm{S}_n} f(\sigma)= \sum_{[\pi]\in\mathrm{S}_n / \mathrm{S}_n^p} \sum_{\sigma\in \mathrm{S}_{n}^{p}} f(\sigma\circ\pi).
\end{equation}
The $n$-particle state \eqref{eqn_fockStatesAnsatz} then takes the form
\begin{equation}
    \fockstate{f} = \frac{1}{|\mathrm{S}_n^p|} \sum_{\sigma\in\mathrm{S}_n^p} \frac{|\mathrm{S}_n^p|}{|\mathrm{S}_n|} \sum_{[{\pi}]\in\mathrm{S}_n / \mathrm{S}_n^p} \sum_{\left\lbrace I\right\rbrace}  \fockSpaceElement{f}_{\left\lbrace I\right\rbrace}    \bigotimes_{l=1}^n \left| \eta(\snf_{c_{\pi(l)}}^{I_{\sigma \circ \pi (l)}})\right),
\end{equation}
where $|\mathrm{S}_n^p| = m_1 ! \dots m_N !$ and $|\mathrm{S}_n| = n!$. Moreover, one can symmetrize the components $\fockSpaceElement{f}_{\{I\}}$ by the permutation $\sigma$ and therefore the states expanded in such a basis can finally be written as
\begin{equation}
\label{eqn_fockStatesAnsatzSimplified}
    \fockstate{f} = \frac{1}{|\mathrm{S}_n^p|} \sum_{\sigma\in\mathrm{S}_n^p} \frac{|\mathrm{S}_n^p|}{|\mathrm{S}_n|} \sum_{[{\pi}]\in\mathrm{S}_n / \mathrm{S}_n^p} \sum_{\left\lbrace I\right\rbrace}  \fockSpaceElement{f}_{\left\lbrace \sigma I\right\rbrace} \bigotimes_{l=1}^n \left| \eta(\snf_{c_{{\pi (l)}}}^{I_{{\pi} (l)}})\right),
\end{equation}
where $\sigma I \equiv \sigma (I_1, I_2, \dots) := I_{\sigma (1)}, I_{ \sigma(2)}, \dots$. We now expand the inner product of $\fockSpace$ accordingly, and show that it takes a very similar form as that in lemma \ref{lemma_resolutionOfDiffInvInnerProduct}.

\begin{lemma}[Inner product on $\fockSpace$]
\label{lemma_resolutionOfFockInnerProduct}
Let $\fockSpace$ be as above, and $\fockstate{f}$, $\fockstate{f'}$ $n$- respectively $n'$-particle states in $\fockSpace$ of the form \eqref{eqn_fockStatesAnsatz} over one-component graphs $c_1,\ldots c_n$ and $c'_1,\ldots c'_{n'}$. 
Then 
\begin{equation}
    \fockProduct{f}{f'} = 0 
\end{equation}
if $n\neq n'$ or $\{[c_k]\}\neq \{[c'_k]\}$, where $[\,\cdot\,]$ denote diffeomorphism equivalence classes, and
\begin{equation}
    \fockProduct{f}{f'} = \frac{1}{|\mathrm{S}_n^p|}\sum_{\pi\in\mathrm{S}_n^p} \frac{|\mathrm{S}_n^p|}{|\mathrm{S}_n|} \sum_{\left\lbrace I\right\rbrace , \left\lbrace I'\right\rbrace} \Bar{\fockSpaceElement{f}}_{\left\lbrace \pi I\right\rbrace} \fockSpaceElement{f}'_{\left\lbrace I' \right\rbrace} \prod_{l=1}^{n}\left(\eta(\snf_{c_{ l}}^{I_{l}}), \eta(\snf_{c_{l}}^{I'_{l}})\right)
\end{equation}
otherwise.
\end{lemma}
\begin{proof}
States with different numbers of particles are orthogonal for any Fock space, and so are states of the form \eqref{eqn_fockStatesAnsatz} in which for each of the summands, one of the constituent one-component states is orthogonal to all one-particle states contained in the other. This is the case if $\{[c_k]\}\neq \{[c'_k]\}$, therefore we have shown the first part of the lemma. In the following we will assume that $n=n'$ and $\{[c_k]\} = \{[c'_k]\}$. Hence through relabeling, without loss of generality we can assume $[c_k]=[c'_k]$ and due to diffeomorphism invariance of the one-particle inner product $c_k=c'_k$ for $k=1,\ldots,n$. Then, using \eqref{eqn_fockStatesAnsatzSimplified}, 
\begin{equation}
    \fockProduct{f}{f'} =
    \begin{aligned}[t]
    \frac{1}{|\mathrm{S}_n^p|}\sum_{\sigma\in\mathrm{S}_n^p}\frac{|\mathrm{S}_n^p|}{|\mathrm{S}_n|}\sum_{[\pi]\in\mathrm{S}_n / \mathrm{S}_n^p} \frac{1}{|\mathrm{S}_n^p|}\sum_{\sigma '\in\mathrm{S}_n^p}\frac{|\mathrm{S}_n^p|}{|\mathrm{S}_n|}\sum_{[\pi ']\in\mathrm{S}_n / \mathrm{S}_n^p} \sum_{\left\lbrace I\right\rbrace , \left\lbrace I'\right\rbrace} &  \Bar{\fockSpaceElement{f}}_{\left\lbrace \sigma I\right\rbrace} \fockSpaceElement{f'}_{\left\lbrace \sigma' I' \right\rbrace} \times
    \\
    & \times\prod_{l=1}^{n}\left(\eta(\snf_{c_{\pi (l)}}^{I_{\pi (l)}}), \eta(\snf_{c_{\pi' (l)}}^{I'_{\pi' (l)}})\right).
    \end{aligned}
\end{equation}
Note that the product of the spin networks is only non-zero when $[\pi] = [\pi']$. This implies that 
\begin{equation}
    \pi'\circ\pi^{-1}\in\mathrm{S}_n^p 
\end{equation}
and hence $\gamma_{\pi(l)}$ is diffeomorphic to $\gamma_{\pi'(l)}$ for every $l$. 
As a result, only one term, as well as the prefactor $|\mathrm{S}_n^p|/|\mathrm{S}_n|$, in the sum $\sum_{[\widetilde{\pi} ']\in\mathrm{S}_n / \mathrm{S}_n^p}$ remains. Lastly, the symmetrization projection applied to the coefficients $\fockSpaceElement{f}$ and $\fockSpaceElement{g}$ need only be applied to either one and as such, the inner product reduces to
\begin{align}
    \fockProduct{f}{f'} & = \frac{1}{|\mathrm{S}_n^p|}\sum_{\sigma\in\mathrm{S}_n^p}\frac{|\mathrm{S}_n^p|}{|\mathrm{S}_n|}\sum_{[\pi]\in\mathrm{S}_n / \mathrm{S}_n^p} \frac{|\mathrm{S}_n^p|}{|\mathrm{S}_n|} \sum_{\left\lbrace I\right\rbrace , \left\lbrace I'\right\rbrace} \Bar{\fockSpaceElement{f}}_{\left\lbrace \pi I\right\rbrace} \fockSpaceElement{g}_{\left\lbrace I' \right\rbrace} \prod_{l=1}^{n}\left(\eta(\snf_{\gamma_{\widetilde{\pi} (l)}}^{I_{\widetilde{\pi} (l)}}), \eta(\snf_{\gamma_{\widetilde{\pi} (l)}}^{I'_{\widetilde{\pi} (l)}})\right) , 
    \\
    & = \frac{1}{|\mathrm{S}_n^p|}\sum_{\pi\in\mathrm{S}_n^p}\frac{|\mathrm{S}_n^p|}{|\mathrm{S}_n|}\sum_{[\widetilde{\pi}]\in\mathrm{S}_n / \mathrm{S}_n^p} \frac{|\mathrm{S}_n^p|}{|\mathrm{S}_n|} \sum_{\left\lbrace I\right\rbrace , \left\lbrace I'\right\rbrace} \Bar{\fockSpaceElement{f}}_{\left\lbrace \pi I\right\rbrace} \fockSpaceElement{g}_{\left\lbrace I' \right\rbrace} \prod_{l=1}^{n}\left(\eta(\snf_{\gamma_{\widetilde{\pi} (l)}}^{I_{\widetilde{\pi} (l)}}), \eta(\snf_{\gamma_{\widetilde{\pi} (l)}}^{I'_{\widetilde{\pi} (l)}})\right) \delta_{I_{\widetilde{\pi} (l)}, I'_{\widetilde{\pi} (l)}} , 
    \\
    & = \frac{1}{|\mathrm{S}_n^p|}\sum_{\pi\in\mathrm{S}_n^p}\frac{|\mathrm{S}_n^p|}{|\mathrm{S}_n|}\sum_{[\widetilde{\pi}]\in\mathrm{S}_n / \mathrm{S}_n^p} \frac{|\mathrm{S}_n^p|}{|\mathrm{S}_n|} \sum_{\left\lbrace I\right\rbrace , \left\lbrace I'\right\rbrace} \Bar{\fockSpaceElement{f}}_{\left\lbrace \pi I\right\rbrace} \fockSpaceElement{g}_{\left\lbrace I' \right\rbrace} \prod_{l=1}^{n}\left(\eta(\snf_{\gamma_{l}}^{I_{l}}), \eta(\snf_{\gamma_{l}}^{I'_{l}})\right) , 
    \\
    \label{eqn_fockInnerProductResult}
    & = \frac{1}{|\mathrm{S}_n^p|}\sum_{\pi\in\mathrm{S}_n^p} \frac{|\mathrm{S}_n^p|}{|\mathrm{S}_n|} \sum_{\left\lbrace I\right\rbrace , \left\lbrace I'\right\rbrace} \Bar{\fockSpaceElement{f}}_{\left\lbrace \pi I\right\rbrace} \fockSpaceElement{g}_{\left\lbrace I' \right\rbrace} \prod_{l=1}^{n}\left(\eta(\snf_{\gamma_{ l}}^{I_{l}}), \eta(\snf_{\gamma_{l}}^{I'_{l}})\right),
\end{align}
which is the $n$-fold product of the spin networks over the components of the graph. Note that in the second line, the only non-trivial terms which remain are ones which have all equal spin network indices. Therefore, since the indices are all matched, the permutations $\widetilde{\pi}$ are dropped from the indices in the next line and as a result, the sum over the permutations $[\widetilde{\pi}] \in \mathrm{S}_{n} / \mathrm{S}_{n}^{p}$ just gives the order of the group and the last line is then obtained. 
\end{proof}

\noindent In the above lemma, we have shown that a Fock space $\fockSpace$ constructed from Hilbert space of the of one component graphs carries an inner product which closely resembles the one obtained on $\hDiff$. Therefore, in the following theorem, an explicit isomorphism is constructed between these two spaces.
\begin{theorem}[An isomorphism between $\fockSpace$ and $\hDiff$]
\label{theorem_FockSpaceIsomorphism}
Consider the map 
\begin{equation}
\label{eqn_iso}
    \fockIsoMap : \diffState{\eta(\Psi_{\gamma})} \longmapsto 
    \frac{c_\gamma}{n!} 
    \sum_{\sigma\in\mathrm{S}_n}\sum_{\left\lbrace I\right\rbrace} \psi_{\left\lbrace I\right\rbrace} \bigotimes_{l=1}^{n} \left| \eta(\snf_{\gamma_{\sigma (l)}}^{I_{\sigma (l)}})\right), \qquad c_\gamma = \sqrt{\frac{n!}{m_1 ! \cdots m_k !}}
\end{equation}
where $\gamma$ is a graph with $n$ unlinked components $\gamma_1, \ldots \gamma_n$, of which subsets of size $m_1, m_2, \ldots, m_k$ are diffeomorphic to each other but not to the rest. 
$p=(m_1, m_2, \ldots, m_k)$ is considered a partition of $n$, and $\mathrm{S}_n^p$ is the subgroup of permutations fixing $p$. Moreover
\begin{equation}
    \Psi = \sum_{I}\psi_{\left\lbrace I \right\rbrace}\prod_{l = 1}^{n}\snf_{\gamma_l}^{I_l}.
\end{equation}
is a cylindrical function in $\mathscr{H}'_\gamma$. Then $\fockIsoMap$
\begin{itemize}
    \item is a well defined map from $\mathscr{H}'_\gamma$ to $\fockSpace = \mathcal{F}_{\text{S}}[\eta (\hilbertSpaceSub{1})]$,
    \item is a partial isometry, 
    \item extends linearly to a unitary map $\hDiff\rightarrow \fockSpace$. 
\end{itemize}
\end{theorem}
\begin{proof}
First we have to show that $\fockIsoMap$ is well defined. To this end, note that according to Lemma \ref{lemma_resolutionOfDiffInvInnerProduct}, the inner product with $\eta(\Psi_\gamma)$ -- and hence $\eta(\Psi_\gamma)$ itself -- is described by 
\begin{itemize}
    \item the list of diffeomorphism equivalence classes $[\gamma_1], \ldots, [\gamma_n]$
    \item corresponding coefficients of group averaged spin networks $\Psi_{I_1,\ldots, I_n}$ with the symmetry properties  
    \begin{equation}
    \label{eqn_sym1}
        \Psi_{I_1,I_2,\ldots I_n}=\Psi_{I_{\pi(1)},I_{\pi(2)},\ldots, I_{\pi(n)}}, \qquad \pi\in \mathrm{S}_n^p, 
    \end{equation}
    and the symmetry imposed on $\Psi_{\{I\}}$ due to the averaging over $\mathrm{H}_\gamma$, the symmetries generated by the graph symmetries of the components (see lemma \ref{lemma_decompositionOfGraphSymmetries}). We could write it as 
    \begin{equation}
    \label{eqn_sym2}
        \Psi_{I_1,I_2,\ldots I_n}=\Psi_{\sigma_1(I_1),\sigma_2(I_2),\ldots, \sigma_n(I_n)}, 
    \end{equation} 
    where $\sigma_i=\sigma_i(\varphi), \varphi\in \mathrm{H}_\gamma$ correspond to permutations of the edge labels. 
\end{itemize}
The assignment of this data is \emph{unique}, up to reordering of graph components and edges: Changing anything else about this data will change the inner product and hence the state. That said, $\fockIsoMap$ is well defined because the image state just depends on the data listed above, and it depends on it in such a way that reordering of the graph component does not change the image. The ordering of the edges is not even explicitly used.  
\startNewParagraph
Next it needs to be shown that $\fockIsoMap$ is a partial isometry. This is ensured by the prefactor $c_\gamma$ in \eqref{eqn_iso}, and lemmas \ref{lemma_resolutionOfDiffInvInnerProduct}, \ref{lemma_resolutionOfFockInnerProduct}, as a short calculation shows. 
\startNewParagraph
Finally we show that $\fockIsoMap$ extends to a unitary map $\hDiff\rightarrow \fockSpace$. Note first that all states in the $n$-particle sector are, by definition, linear combinations of states of the form \eqref{eqn_fockStatesAnsatz}. The latter can be simplified to \eqref{eqn_fockStatesAnsatzSimplified}, hence they are described precisely by a list of diffeomorphism equivalence classes $[\gamma_1], \ldots, [\gamma_n]$ of unlinked components, together with coefficients $\fockSpaceElement{f}_{\left\lbrace \sigma I\right\rbrace}$ with the symmetry properties \eqref{eqn_sym1}, \eqref{eqn_sym2}. Hence the map 
\begin{equation}
    \fockIsoMap: \eta(\mathscr{H}_n') \longrightarrow  \text{Sym} \bigotimes_{k=1}^{n} \eta(\mathscr{H}_1'), 
\end{equation}
where $\text{Sym}$ is the symmetrization of the tensor product, is onto and hence unitary. Finally, we can again extend by linearity to sectors with different $n$ and hence, to all of $\hDiff$. The result is obviously onto. Sectors with different $n$ are orthogonal, both, in $\hDiff$ and in $\fockSpace$. So the extended map still preserves the inner product and hence 
$\fockIsoMap: \hDiff\rightarrow \fockSpace$ is unitary. 
\end{proof}
\noindent It has then been shown that the diffeomorphism invariant Hilbert space of LQG can be viewed as a Fock space constructed from the diffeomorphism invariant Hilbert space of states over one-component graphs.  This isomorphism introduces the Fock structure one typically encounters and is familiar with in QFTs to the context of LQG. This opens new perspectives on the states in $\hDiff$. One can interpret quantum geometry to have an \textit{atomic} structure,  where every graph component plays the role a of an atom.\footnote{Perhaps an even more apt picture would be the one-component states as \emph{molecules}, which themselves are composed of atoms (the vertex states).}   
For $\gamma \in \Gamma_n$ one has a spacetime composed of $n$ such particles. 
\startNewParagraph
Note however that this Fock space structure on $\hDiff$ may not be the only one possible. In other recent work, a Fock space structure has been given to subspaces of $\hDiff$, by treating excitations at the vertices of the graph 
as single particle states \cite{Assanioussi:2018zit,Assanioussi:2020fsz}, for example. 
\startNewParagraph
In the following sections, we will make use of the new Fock structure for making contact with group field theory. We will also try to illuminate the geometric nature of the atoms, i.e., the states over one-component graphs, in the following.  

\section{Applications and interpretation of the Fock space structure}
\label{section_applications}

In this part of the paper we explore the Fock space structure in physically relevant contexts. We start by studying coherent states to emphasize the fact that one can write these states in a simple manner using the Fock structure. Some examples include coherent states with a volume expectation value which is peaked around a certain value. We also consider states with an infinite volume expectation value. Next, we compare the states one obtains in $\fockSpace$ with the multi-particle condensates used in GFT. We show that although the manner of which the states are derived is different, one arrives at the same end result and thus establishing a quantitative connection between the two approaches. Lastly, we gain some insight to how one understands the geometry of the component Fock states by studying geometric observables associated to leaves of a foliation in the quantum theory and in the classical theory. 

\subsection{Coherent states}
\noindent One of the important tools to study quantum systems are coherent states. Using them, one investigates the effects of quantum fluctuations on the underlying corresponding classical description of the quantum system of interest. The component Fock states are a superposition of states which are cylindrical over a graph with infinitely many components. This is something that is not easily done in the context of $\hDiff$. 
\startNewParagraph
We have on the Fock space $\fockSpace$ creation and annihilation operators which obey the standard canonical commutation relations. Consider then a single particle state $\fockstate{g}$ which is taken to be normalized to one. For conciseness we also write 
\begin{equation}
    \creationOperator := \creationOperator_{\fockstate{g}}, \quad \annihilationOperator := \annihilationOperator_{\fockstate{g}}.
\end{equation}
We note that $\annihilationOperator \fiducialState = 0$ where $\fiducialState$ is simply the (image of the) Ashtekar-Lewandowski vacuum in $\fockSpace$.   
\startNewParagraph
Given the above mathematical layout, one easily constructs coherent states in $\fockSpace$ which take the form
\begin{equation}
\label{eqn_coherent}
    \fockstate{f} := e^{-|c|^{2} / 2} e^{c \creationOperator} \fiducialState.
\end{equation}
It is easy to see that such states are normalized. One way to investigate the properties of such states is to look at the action of some observable on them. The volume operator constitutes a good candidate to do so. Albeit its complicated eigenvalue spectrum \cite{Ashtekar:1997fb}, it offers a direct physical insight into the geometric structure of a spacetime described by such Fock structure. It is well known that there exists eigenstates of the volume operator $V$ defined on $\hDiff$ \cite{Ashtekar:1997fb}. 
We defined the volume operator $\fockOperator{V}$ on $\fockSpace$ as
\begin{equation}
    \fockOperator{V} := \fockIsoMap V \fockIsoMap^{-1} 
\end{equation}
Taking the single particle state $\fockstate{g}$ to be an eigenstate of $\fockOperator{V}$ with an eigenvalue of $\lambda$, then looking at the expectation value of such an operator results in
\begin{align}
\expectation{\fockOperator{V}}_{\fockstate{f}} & = e^{-|c|^{2}} \sum_{m} \sum_{n}\frac{1}{m!}\frac{1}{n!} \dualFockState{0} \Bar{c}^{m} \annihilationOperator^{m} \fockOperator{V} c^{n} (\creationOperator)^{n} \fockstate{0} ,
\\
& = \lambda e^{-|c|^{2}} \sum_{m} \sum_{n} n \frac{\Bar{c}^{m}}{m!} \frac{c^{n}}{n!} n! \delta_{m, n} ,
\\
& = \lambda |c|^2 .
\end{align}
Moreover, it is straightforward to see that probabilities of eigenvalues of $\fockOperator{V}$ in this state follow Poisson statistics. This shows how simple it is to construct states with interesting properties in $\fockSpace$. 
\startNewParagraph
As another example, one can consider states with a diverging volume expectation value. Such states which are not in the domain of $V$ must exist since $V$ is unbounded. They might be interesting for the description of non-compact spatial geometries in the quantum theory. We can use the Fock structure to easily construct such states. To do so, consider once again an arbitrary single particle state $\fockstate{l}$ and let it be normalized and, furthermore, be a volume eigenstate with eigenvalue $\lambda$. Now, construct from the state $\fockstate{l}$ the $m$ fold tensor product state 
\begin{equation}
    \fockstate{l_{m}}:= \fockstate{l}^{\otimes m}.
\end{equation}
We can then consider linear combinations of such tensor product states. That is, for $c_{k} \in \mathbb{C}$ we define
\begin{equation}
    \fockstate{L} := \sum_{n = 0}^{\infty} c_{n}\fockstate{l_{n}}.
\end{equation}
One can easily check that in general the action of the volume operator $\fockOperator{V}$ is simply the second quantization of the volume operator $V$ on the one particle Hilbert space $\eta(\mathscr{H}'_1)$.  
Therefore
\begin{equation}
    \fockOperator{V} \fockstate{L} = \sum_{n = 0}^{\infty} c_{n} \fockOperator{V} \fockstate{l_{n}} = \lambda \sum_{n = 0}^{\infty} n c_{n} \fockstate{l_{n}}.
\end{equation}
Such states, which are linear combinations of tensor product states, contribute to an arbitrarily high value to the spectrum of the volume operator. Since every $\fockstate{l}$ is normalized to one, it is clear that $\normsquared{\fockstate{l_{m}}} = 1$ too. We now choose the coefficients $c_{k} \in \mathbb{C}$ such that $c_0 = 0$ and $c_{n > 0} = 1/n$. In such a case, the norm of $\fockstate{L}$ then becomes
\begin{equation}
    \normsquared{\fockstate{L}} = \sum_{n = 0}^{\infty} |c_{n}|^{2} = \frac{\pi^{2}}{6}
\end{equation}
A rather straightforward calculation of the expectation value of the volume operator shows that
\begin{equation}
    \expectation{\fockOperator{V}}_{\fockstate{L}} = \frac{1}{\normsquared{\fockstate{L}}}\sum_{n = 0}^{\infty} n \lambda |c_{n}|^{2} = \lambda \frac{6}{\pi^{2}}\sum_{n = 1}^{\infty} \frac{1}{n}, 
\end{equation}
which is nothing other than the diverging harmonic series. Thus we have constructed an easy example in which large volume eigenstates (based on large graphs) are entering with a weight sufficient to make the volume of $\Sigma$ infinite. 

\subsection{Connection to group field theory}

In this section, the Fock states are compared to the multi-particle condensate states one encounters in GFT. We show that although the two approaches differ in the manner of which the states are constructed, the resulting states are quantitatively the same.
\startNewParagraph
We start by outlining the process in which one constructs multi-particle condensates in the context of GFT. We follow the work done by Gielen \etal \cite{Gielen:2013kla,Gielen:2019kae,Gielen:2016dss} as a reference. One starts from a bosonic field, denoted by $\phi(g_I) := \phi(g_1, g_2, g_3, g_4)$ where $g_{I} \in \mathrm{SO}(4)$, which one can expand as a field operator defined on every vertex of a given graph $\gamma$ such that
\begin{equation}
    \Hat{\phi}(g_{I}) := \sum_{v\in V(\gamma)} \phi_{v}(g_{I})a_{v}.
\end{equation}
Here, $a_{v}$ denotes an annihilation operator, one whose spectrum is bounded from below by the condition that for a vacuum state $\ket{0}$ then $a_{v}\ket{0} = 0$. $a^{\dagger}_{v}$ denotes the corresponding creation operator. The geometry of the states can be understood as follows. A creation operator $\phi^{\dagger}(g_I)$ creates quantum geometries which are interpreted as tetrahedron with a geometry which can be obtained by the parallel transport of the elements $g_{I}$ of the gravitational $\mathrm{SO}(4)$ connection along links dual to its faces. One can then associate bivectors to each face of the tetrahedron, as defined in \cite{Gielen:2013kla}, such that
\begin{equation}
    B_{\Delta_I}^{AB} \sim \int_{\Delta_I} e^{A}\wedge e^{B},
\end{equation} 
where $e$ is a cotetrad field which encodes the simplical geometry. However, the bivectors must uphold two conditions to ensure the proper geometric interpretation \cite{Gielen:2013kla}. Firstly, the simplicity constraint, which can be obtained by imposing the condition that for every $(h_{I}) \in\mathrm{SO}(3)^3$
\begin{equation}
    \phi(g_1 , g_2 , g_3 , g_4) = \phi(g_1 h_1 , g_2 h_2 , g_3 h_3 , g_4 h_4) =: \phi(g_I h_I)
\end{equation}
where a particular $\mathrm{SO(3)}$ subgroup of $\mathrm{SO}(4)$ was chosen for the action of the former on the latter. 
Therefore, the field $\phi$ effectively depends on four copies of $\mathrm{SO}(4)/\mathrm{SO}(3) \sim \mathrm{SU}(2)$. The second condition is gauge invariance which can be imposed as
\begin{equation}
    \forall h\in\mathrm{SO}(4) : \phi(g_1 , g_2 , g_3 , g_4) = \phi(g_1 h , g_2 h , g_3 h , g_4 h)
\end{equation}
Consequently, one can then define a naturally gauge invariant two-particle condensate state as
\begin{equation}
\label{eqn_gftTwoParticleState}
    \ket{\xi} := e^{\hat{\xi}}\ket{0} \qquad , \qquad \hat{\xi} := \frac{1}{2}\int dg_1 \cdots dg_4 \int dh_1 \cdots dh_4 \, \xi(g_I h_{I}^{-1})\hat{\phi}^\dagger (g_I) \hat{\phi}^\dagger (h_I)
\end{equation}
for some function $\xi$ on $\mathrm{SO}(4)^4$. In the spin network picture one sees that the operator $\hat{\phi}^\dagger (g_I) \hat{\phi}^\dagger (h_I)$ creates a state which has a geometry which corresponds to two 4-valent vertices with four open ends. The open ends are then connected via the integrations enforcing gauge invariance. Note that one can define single- (or $n>2$) particle condensate states by adapting the number of field operators in the expression.
\startNewParagraph
The definition of $\ket{\xi}$ looks qualitatively similar to that of the coherent state \eqref{eqn_coherent}. However, the similarity is more than just a qualitative one. Let us investigate how 
$\hat{\xi}$ can be simplified if $\xi$ is chosen to be a gauge spin network state
\begin{equation}
    \xi(g_I) := \left(\iota_0\right)_{b_1 \dots b_4} \left(\iota_1\right)^{a_1 \dots a_4} \pi_{j_1}(g_1)_{a_1}{}^{b_1} \dots \pi_{j_4}(g_4)_{a_4}{}^{b_4}.
\end{equation}
This is a state which has two vertices and 4 edges connecting them. Here, $\iota$s and $\pi$s denote the intertwiners and spin labeled SU(2) representations respectively. For this $\xi$, the operator $\hat{\xi}$ defined above becomes
\begin{multline}
        \hat{\xi} = \int dg_1 \cdots dg_4 \, \int dh_1 \cdots dh_4 \sum_{c_I} \left(\iota_0\right)_{b_1 \dots b_4} \left(\iota_1\right)^{a_1 \dots a_4} \pi_{j_1}(g_1)_{a_1}{}^{c_1}\pi_{j_1}(h_{1}^{-1})_{c_1}{}^{b_1} \dots \pi_{j_4}(g_4)_{a_4}{}^{c_4} \times 
        \\
        \times \pi_{j_4}(h_{4}^{-1})_{c_4}{}^{b_4}  \hat{\phi}^\dagger (g_I)\hat{\phi}^\dagger (h_I).
\end{multline}
In order to evaluate this integral one can separate the variables $g_I$ and $h_I$ such that
\begin{multline}
    \hat{\xi} = \left(\iota_0\right)_{b_1 \cdots b_4} \left(\iota_1\right)^{a_1 \cdots a_4} \sum_{c_I} \left(\int dg_1 \cdots dg_4 \, \pi_{j_1}(g_1)_{a_1}{}^{c_1} \cdots\pi_{j_4}(g_4)_{a_4}{}^{c_4}\hat{\phi}^\dagger (g_I)\right) \times \\
    \times \left(\int dh_1 \cdots dh_4 \,  \pi_{j_1}(h_1^{-1})_{c_1}{}^{b_1} \cdots\pi_{j_4}(h_4^{-1})_{c_4}{}^{b_4}\hat{\phi}^\dagger (h_I)\right).
\end{multline}
This can be rewritten as a gauge invariant state in terms of pairs of creation operators which in effect act by connecting free edges which emanate from two vertices which have the intertwiners $\iota_{0}$ and $\iota_{1}$ associated to them. Therefore, after evaluating the integrals
\begin{equation}
    \hat{\xi} = \left(\iota_0\right)_{b_1 \dots b_4} \left(\iota_1\right)^{a_1 \dots a_4} \sum_{c_I} \left(a^\dagger_{j_1}\right)_{a_1}{}^{c_1}\dots\left(a^\dagger_{j_4}\right)_{a_4}{}^{c_4} \left(a^\dagger_{j_1}\right)_{c_1}{}^{b_1}\dots \left(a^\dagger_{j_4}\right)_{c_4}{}^{b_4}.
\end{equation}
In the language of LQG, we see that the $\hat{\xi}$ acts by connecting the free edges attached to the two vertices. In this specific example, the result is a 4-valent two vertex graph. The same prescription applies to constructing graphs of any number of vertices and edges. We can now build the analogue of the condensate state $\ket{\xi}$ in $\fockSpace$. To this end, let $\xi_\gamma \in \mathscr{H}_\text{kin}$ be a spin net based on a two vertex, four edge dipole graph $\gamma$ embedded in $\Sigma$.   
 $\ket{\xi}$ can then be identified with the coherent state 
\begin{equation}
    | \xi ):= e^{\creationOperator_{\eta(\xi_\gamma)}}  \fiducialState \quad \in \fockSpace. 
\end{equation}
This state will have the same physical properties as $\ket{\xi}$ of \eqref{eqn_gftTwoParticleState}, now with respect to observables on $\fockSpace$ that can be obtained via second quantization from operators on $\eta(\mathscr{H}'_1)$ or more generally via $\fockIsoMap$ from those on $\hDiff$. 
We therefore see that the process in which the resulting state is constructed differs between the two approaches but one can quantitatively identify one with the other. 

\subsection{Towards a geometric interpretation of multi-particle states}

\noindent This section aims to elucidate on the geometric interpretation of eigenstates of the number operator on $\fockSpace$, i.e., $n$-particle states with fixed $n$. Does $n$ have a geometric corelate? 
To shed light on this question, we will consider area and volume operators associated to foliations of $\Sigma$ into leaves of topology $S^2$.  We begin by looking at the quantum picture in which we consider  area operators which act on a spin network state that would end up in the two-particle sector under group averaging. We are interested in qualitative behaviour of the expectation values as we go through the leafs of the foliation. 
Next, we search for a classical geometry that would lead to similar results as obtained from the quantum picture. Once again the area associated to foliations are studied. We draw parallels between the two regimes and provide a geometric interpretation based on the classical theory. We then outline a similar analysis based on the volume operator and conclude by a short discussion on the imposition of diffeomorphism invariance on the states.

\subsubsection{Area in a foliation: the quantum picture}
\label{subsec_expectationValueOfTheAreaOperator}
\noindent 

\noindent In this section, we study the global behavior of some area operator in the quantum picture. We consider the standard area operator found in the literature of LQG \cite{Ashtekar:1996eg}. As this is more of a qualitative study, we note that we do not explicitly write out the definition of such an operator (see \cite{Ashtekar:1996eg} for details). We do note that the area operator can in principle be well defined on $\eta(\hilbertSpaceSub{1})$ if one can fix the surface by reference to physical fields. It would then also transfer to $\fockSpace$. However, this is not done here. 
\startNewParagraph
We start by specifying the graphs that are embedded in $\spatialSubMfd$. For the sake of simplicity, we consider a graph $\gamma \in \Gamma_{2}$. 
Next, we consider a foliation of $S_t$ of $\Sigma$ by $S^2$s, starting from some point $S_0$ (i.e. a degenerate leaf) in $\Sigma$.\footnote{We could for example construct such a foliation by using geometric 2-spheres $S_t :=S(R(t))$ with radii $R(t)$ such that $R(t_{1}) < R(t_{2})$ for $t_{1} < t_{2}$, with respect to some fiducial metric.} 
We require the surfaces of the foliation to be well behaved. Specifically, for any edge $e$ of a graph $\gamma \hookrightarrow \spatialSubMfd$ the surface $S_t$ is called well behaved if it intersects with $e$ a finite number of times, which is taken to be less than some arbitrary number $n_{\mathrm{max}}$, as shown in Figure \ref{fig_wellBehavedSurface}.
\begin{figure}[ht]
    \centering
    \includegraphics[scale = 0.04]{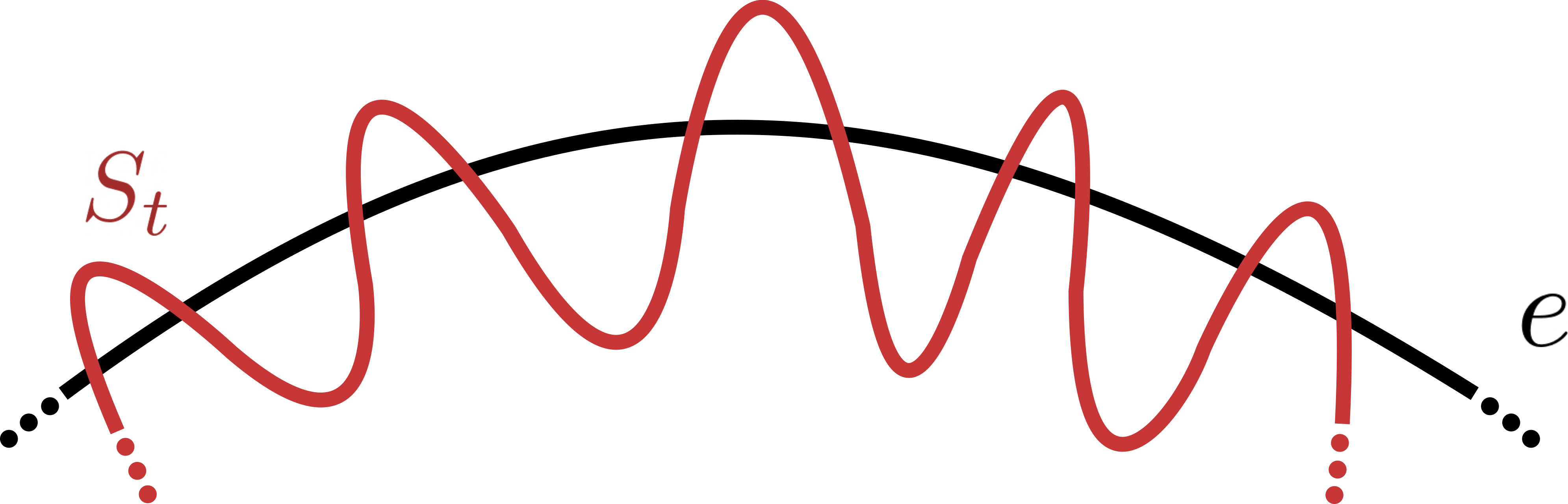}
    \caption{An illustrative picture in 1-dimension of an edge $e$ of some graph $\gamma$ and a surface $S_t$ from some foliation of the spacetime hypersurface $\spatialSubMfd$. We see that the surface $S_t$ intersects with $e$ only a finite amount of times. We call such a surface a well-behaved surface.}
    \label{fig_wellBehavedSurface}
\end{figure}
\noindent This restriction ensures the non-divergence of the expectation value of the area operator. 
Consider, for example, the two different foliations depicted in Figure \ref{fig_quantumAreaOperatorSpectrum}. 
\begin{figure}[ht]
    \centering
    \includegraphics[scale = 0.09]{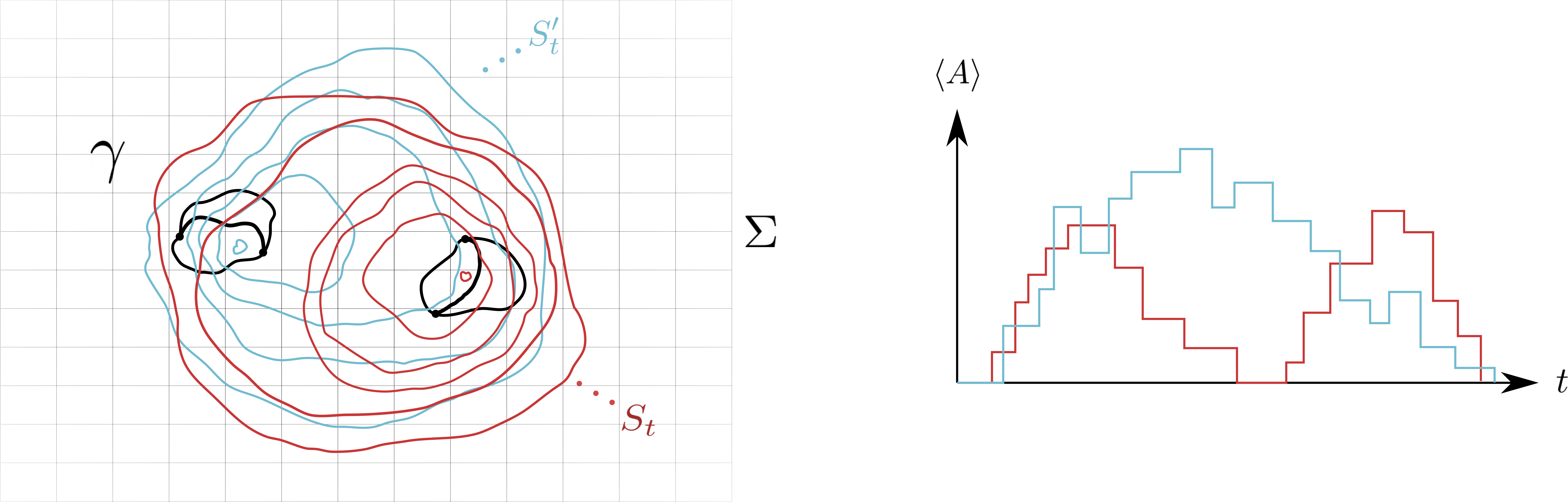}
    \caption{On the left, we see a portion of the spatial hypersurface $\spatialSubMfd$ of spacetime. It is to be noted that it extends to infinity in all directions and for illustrative purposes, the dimensions have been reduced to 2. In black, we see the two unlinked components of the graph $\gamma$ and two different foliations of $\spatialSubMfd$ in red and blue. A qualitative plot of the expectation value of the area of each of the surfaces of both the foliations is shown on the right.}
    \label{fig_quantumAreaOperatorSpectrum}
\end{figure}
\startNewParagraph
In the figure, the unlinked components of the graph $\gamma$ are shown in black. Two distinct foliations of $\spatialSubMfd$ are shown in red and blue where in the former, the surfaces foliate $\spatialSubMfd$ in such a way that it does so component by component. That is, at some point, the first graph component must be fully enveloped by some surface $S_t$ of the foliation before it begins to envelop the next component. On the other hand, in the foliation shown in blue, this condition is not there. This corresponds to two different global behaviours of the expectation values of the area of these surfaces.
Both are possible, depending on the choice of foliation. 
\startNewParagraph
The number of intersections of a given edge of the graph with some surface $S_t$ is finite. This is due to the fact that the surface is well-behaved and each point of intersection between any edge of the graph and the surface lies only within one surface of the foliation. Irrespective of the structure of the given graph, the expectation value of the area in any foliation $S_t$ starting from an arbitrary point $x \in \spatialSubMfd$ and a foliation parameter $t = t_{0}$ will start at zero. After some cutoff value $t = t_{c}$, all the intersections between the graph and the surfaces of the foliation would have been accounted for and consequently, the expectation value of the area operator returns to and remains at zero again. Between $t_{0}$ and $t_{c}$, the expectation value of the area is non-zero and one can infer the geometric structure of the graph $\gamma$ from it. It is of particular significance that there exist foliation with a double hump structure i.e. the area expectation value returning to zero and then growing again.
\startNewParagraph
One sees from Figure \ref{fig_quantumAreaOperatorSpectrum} that the choice of foliation gives an insight and a physical interpretation to the geometric structure of the graphs. It is tempting to view the different components of $\gamma$ as distinct parts of space. However, this can be further explored by searching for a classical analog to the results obtained in the quantum picture.

\subsubsection{Area in a foliation: the classical picture}
\label{subsec_classicalPictureOfTheAreaOperator}
\noindent  

\noindent In this section, we search for a classical analog for the result of the area expectation value spectrum obtained in the quantum regime. Recall that spacetime is taken to be topologically decomposed as $\spacetimeMfd = \mathbb{R}\times\spatialSubMfd$. The search for a classical analog therefore reduces to a search for a geometry for $\spatialSubMfd$ which can be foliated and give a similar picture with respect to the area values associated with the foliation.
\startNewParagraph
As a first example, we consider $\spatialSubMfd = \mathbb{R}^3$ and equipped with some Riemannian metric. The action of the area operator is one which gives the area of a bounded region in $\spatialSubMfd$. Let $x$ be an arbitrary point in $\spatialSubMfd$ and foliate $\spatialSubMfd\setminus\curlybrackets{x}$ with a generic foliation the leaves of which are topological two-surfaces denoted by $S_t$ with the foliation parameter $t \in \left[ 0, \infty\right)$. The area operator then measures the area bounded within the region $S_t$ for every $t$.
\begin{figure}[ht]
    \centering
    \includegraphics[scale = 0.1]{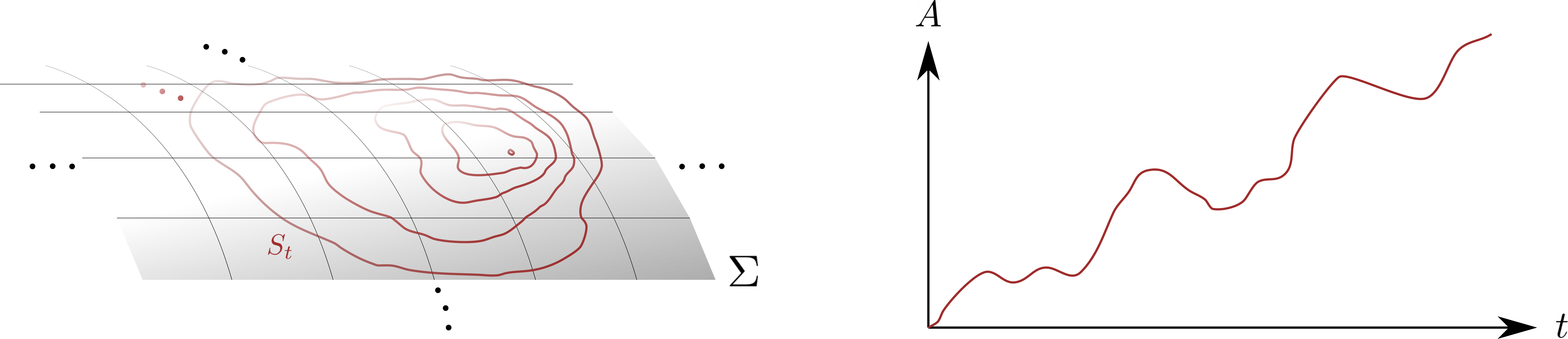}
    \caption{A region of $\spatialSubMfd = \mathbb{R}^3$ is shown on the left. The dimension is lowered by 1 for illustrative purposes. Note that $\spatialSubMfd$ extends to infinity in all directions. The first few surfaces $S_t$ of the foliation of $\spatialSubMfd$ are shown in red. On the right, a qualitative plot of the area of each of the surfaces $S_t$ is shown. A global behaviour can be seen where as $t$ grows, the area of the corresponding surface $S_t$ also grows.}
    \label{fig_r3AreaOperatorSpectrum}
\end{figure}
\startNewParagraph
 Note that the representation shown in Figure \ref{fig_r3AreaOperatorSpectrum} is reduced by 1 dimension purely for illustrative purposes. One can easily see that in such a spatial hypersurface, foliated by well behaved surfaces $S_t$, the area operator will have an overall increasing value as the foliation parameter $t$ increases, as can be illustratively seen in Figure \ref{fig_r3AreaOperatorSpectrum}. Here, by well behaved we mean that we require the leaves of the foliation to have an extrinsic curvature which is less than some fixed $K_{\mathrm{max}}$. This ensures that the leaves do not behave in such a way that allows the area to grow without a bound, for a bounded parameter $t$. This requirement is in place for all foliations that we will consider in the classical picture. Recall that a similar restriction was imposed in the quantum regime, therefore mirroring and further supporting why such a restriction is imposed. 
\begin{figure}[ht]
    \centering
    \includegraphics[scale = 0.11]{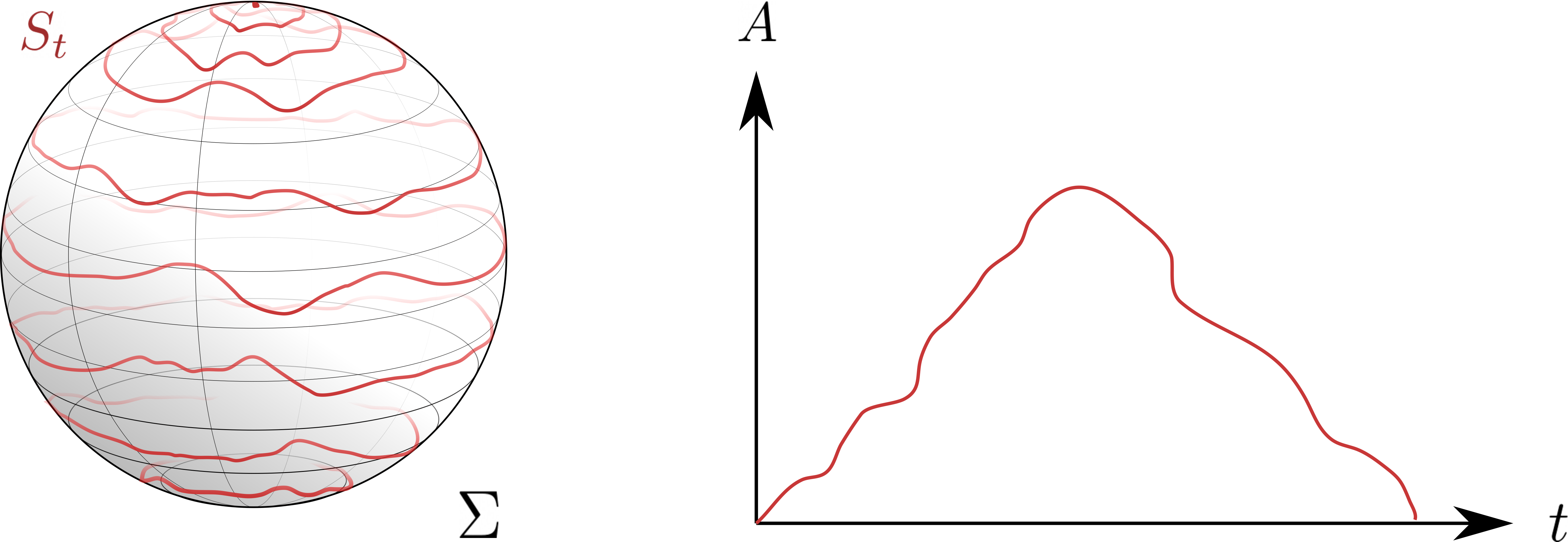}
    \caption{On the left, a spatial hypersurface that has the topology of $\spatialSubMfd = S^3$ is illustrated. The dimensions are lowered by 1 for illustrative purposes. The foliation of such a spatial hypersurface by surfaces $S_t$ can be seen in red. On the right, a qualitative plot of the area of the surfaces $S_t$ is shown as $t$ grows to infinity. We see that the area of the surfaces $S_t$ start from and end at zero as shown on the right.}
    \label{fig_sphereAreaOperatorSpectrum}
\end{figure}
\startNewParagraph
One can also consider the spatial hypersurface to be geometric $\spatialSubMfd = S^3$ where for illustrative purposes, as shown in Figure \ref{fig_sphereAreaOperatorSpectrum}, we again draw the corresponding picture in 1 dimension lower. Similar to the previous case of $\spatialSubMfd = \mathbb{R}^3$, considering appropriate foliations, one sees that the value of the area operator will follow a specific trend. Namely, the area operator will always start from zero and go back to zero. This is to be expected as we foliate the sphere from one arbitrarily chosen pole to another, regardless of the surface of foliation of choice. The last example we consider for the classical picture is shown in Figure \ref{fig_twoBubblesAreaOperatorSpectrum} where yet again the dimensions are lowered by 1 for illustrative purposes. 
\begin{figure}[ht]
    \centering
    \includegraphics[scale = 0.1]{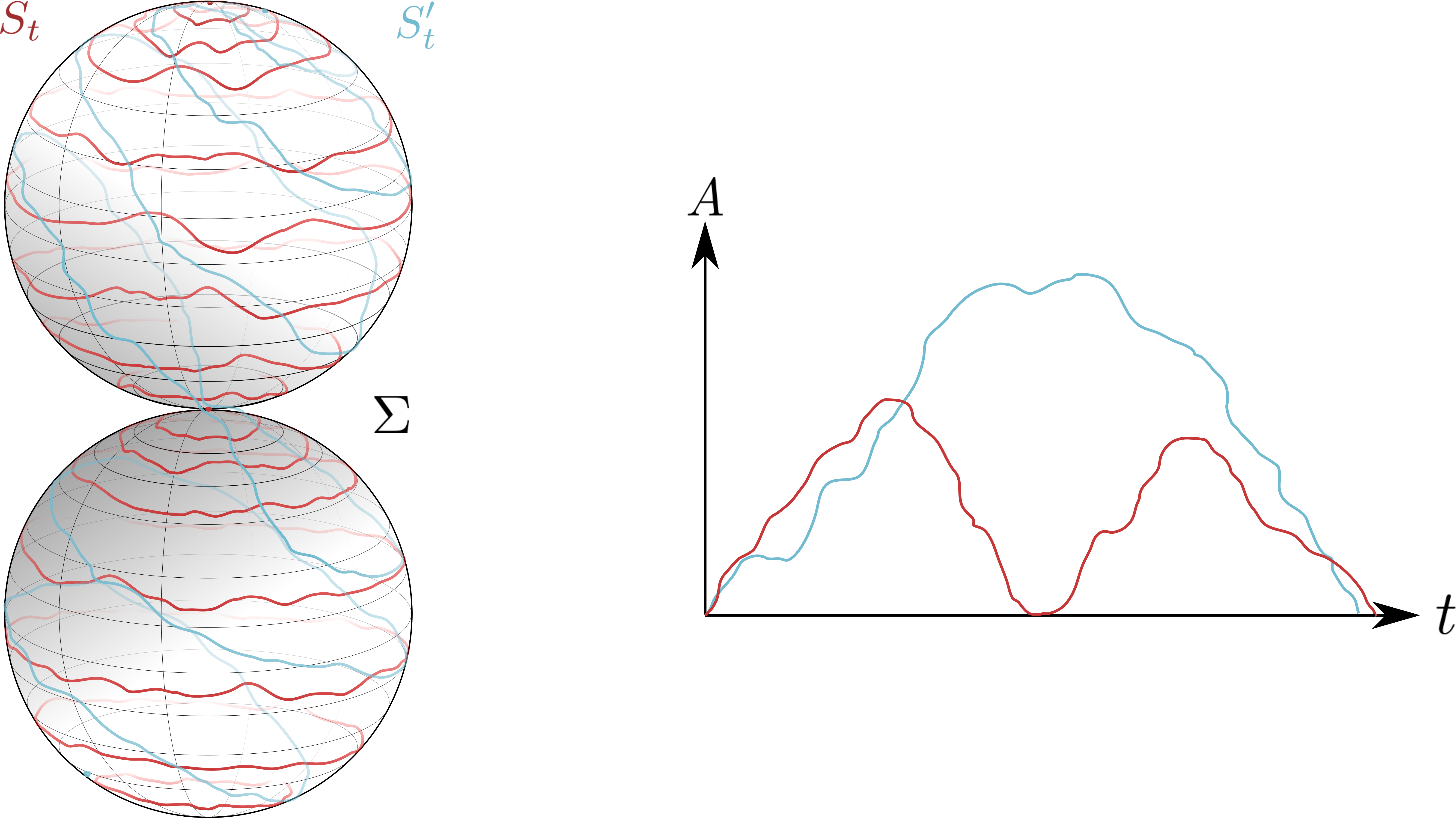}
    \caption{On the left, we see a visual representation of the spatial hypersurface $\spatialSubMfd$ of spacetime with two different foliations shown in red and blue. The former precisely passes through the gluing point of the two balls while the latter does not. On the right, we see a qualitative plot for the area of the surfaces arising from the two foliations. In the one shown in red, we can see that each peak represents one of the balls and the existence of a foliation with such an area behavior gives us an insight into the geometry of $\spatialSubMfd$}
    \label{fig_twoBubblesAreaOperatorSpectrum}
\end{figure}
\startNewParagraph
One can foliate the hypersurface $\spatialSubMfd$ in different ways. Most noticeably, one can choose a foliation such that it passes through the \quotes{gluing} point of the two balls (that is, there exists a surface $S_{t_k}$ in the foliation such that it is precisely the gluing point of the two spheres). This way, we observe a dip in the value of the area operator and then a gradual increase once again and therefore resembling the spectrum we would get for two spheres. Nevertheless, as shown in the second foliation in blue, we can have a different spectrum in which this effect is not seen by choosing a different foliation. While both are valid, the fact that the former exists gives us insight into the geometrical structure of the hypersurface $\spatialSubMfd$ and how one can visualise it.
\startNewParagraph
Moreover, this result in the classical picture closely resembles what was obtained in the quantum regime. This suggests different ways to understand the geometry of such graphs $\gamma \in \Gamma_{n}$, one of which is that these graphs resemble \quotes{bubbles} of space, essentially distinct chunks of space. This leaves room for further interpretation such as considering these bubbles to be parts of the same universe, or independent universes on their own. While certainly such an interpretation is intriguing, it should be taken lightly as it is one of the ways to understand such objects.

\subsubsection{Volume in a foliation in the quantum theory}
\label{subsec_theSpectrumOfTheVolumeOperator}
\noindent

\noindent In this section, we provide an outline for an analogous study to that of the area operator but now considering the volume operator. In the classical picture, the volume operator now measures the volume of the region bounded by a foliation surface $S_t$ in the spatial hypersurface $\spatialSubMfd$. If one considers the same three geometries of $\spatialSubMfd$ as done in the case of the area in a classical picture, it is not difficult to see that a qualitative spectrum for the volume in a classical picture will be one which will always grows, either indefinitely or not, from zero and never goes back to it.
\startNewParagraph
Moving to the quantum picture, we recall that in LQG, for a region $R$ in $\spatialSubMfd$ then the volume operator  is given as (for details see \cite{Ashtekar:1997fb})
\begin{equation}
    V(R) \sim \sum_{v\in R} V_v,
\end{equation}
where as such $V_v$ is the volume at the vertex $v$. We again consider a simple two component graph as in the case of the area operator. Since the component Fock states on such a graph are composed from spin-network functions, the volume operator will have a known, albeit complicated, discrete eigenvalue spectrum \cite{Ashtekar:1997fb,Brunnemann:2007ca}. But that does not hinder our considerations, as we can either chose a state that is an eigenstate of the vertex-volumes, or we just consider volume expectation values. 
Note that unlike the situation for the area operator, the expectation values of the volume operator are summed over the regions in a foliation of the spatial hypersurface $\spatialSubMfd$. The volume operator has an eigenvalue of zero for a vacuum state. Moreover, one can again consider an arbitrary point $x$ on $\spatialSubMfd$ and construct the foliation ``around" $x$. That is, one requires that $x$ lies in a neighborhood $U$ which is in the region $R_{t_1}$ which in turn lies inside $R_{t_2}$ and so on, where $t$ is once again a foliation parameter ranging from $0$ to $\infty$. In doing so, the eigenvalue spectrum of the volume operator for a given Fock state is ensured to start from zero. Since the spectrum is additive it only increases with each region $R_{t_k}$.
\startNewParagraph
As a result, one sees a similar behavior when compared to the area operator. For each component of the graph $\gamma \hookrightarrow \spatialSubMfd$, the eigenvalue spectrum of the volume operator increases to a certain value and remains constant at that value until the foliation surfaces $R_t$ encounter another component. It is not so clear how the structure of the geometry can be inferred in this case. 

\subsubsection{Imposing diffeomorphism invariance}
\noindent

\noindent  A more careful analysis would impose diffeomorphism invariance appropriately, a task which is not straightforward. There are several ways to define such area and volume operators on $\hDiff$. One way to do so is to  make use of scalar reference fields permeating spacetime \cite{Brown:1994py,Rovelli:1993bm,Giesel:2012rb,Lewandowski:2015xqa,Giesel:2018opa,Giesel:2020xnb}.  Surfaces could be defined as surface at which a reference field takes a certain prescribed value. It is possible that an analysis based on such diffeomorphism invariantly defined foliations would lead to different conclusions than our analysis based on diffeomorphism non-invariant states.

\section{Discussion and outlook}
In this work we have shown that under certain assumptions (see Assumptions \ref{AssumptionsForTheFockSpaceIsomorphism}) the diffeomorphism invariant Hilbert space $\hDiff$ of LQG is isomorphic to a symmetric Fock space $\fockSpace$ which brings to light the graph component dependent structure of $\hDiff$. This was done by first choosing $\spatialSubMfd = \mathbb{R}^{3}$, which then enables us to decompose any given graph $\gamma \hookrightarrow \spatialSubMfd$ to its unlinked components $c_{k} \in \Comp{\gamma}$. After mapping the action of the  group of graph symmetries for such graphs to the symmetric group $\mathrm{S}_n$, $\graphSymmetries_{\gamma}$ was then decomposed into the subgroup $\mathrm{H}_{\gamma}$ permuting diffeomorphic graph components, and its cosets $\widetilde{\graphSymmetries}_{\gamma}$. Using these structures, a simple formula for the inner product on $\hDiff$ was given in lemma \ref{lemma_resolutionOfDiffInvInnerProduct}. A Fock space $\fockSpace$, which we called the component Fock space, was constructed over $\eta (\hilbertSpaceSub{1})$. To obtain an isomorphism between $\hDiff$ and the Fock space $\fockSpace$, the latter had to be a symmetric Fock space. This fits well with the bosonic nature of the gravitation field. Investigating the inner product on $\fockSpace$ facilitated the construction of an explicit isomorphism $\fockIsoMap$, 
\begin{align}
\fockIsoMap :  \hDiff &\longrightarrow \fockSpace, \\
\diffState{\eta(\Psi_{\gamma})} &\longmapsto 
    \frac{c_\gamma}{n!} 
    \sum_{\sigma\in\mathrm{S}_n}\sum_{\left\lbrace I\right\rbrace} \psi_{\left\lbrace I\right\rbrace} \bigotimes_{l=1}^{n} \left| \eta(\snf_{\gamma_{\sigma (l)}}^{I_{\sigma (l)}})\right), \qquad c_\gamma = \sqrt{\frac{n!}{m_1 ! \cdots m_k !}}
\end{align}
thus establishing the isomorphism between the full diffeomorphism invariant Hilbert space of LQG and $\fockSpace$.
We note again that this map was obtained by making assumptions on the class of diffeomorphims at hand. In which categories (smooth, semianalytic, \ldots) these assumptions are justified is a question that remains to be answered.
\startNewParagraph
We have then touched upon some examples of the application of such a structure. Namely, we showed that states with interesting properties, such as states which do not fall within the domain of the volume operator, or a certain type of coherent state can be easily constructed using the Fock structure $\fockSpace$. This was facilitated by making use of the creation and annihilation operators on $\fockSpace$. 
\startNewParagraph
Additionally, we demonstrated that the coherent states constructed using the Fock structure have a close connection to some of the condensate states obtained in GFT.  For concreteness we considered a two-particle condensate in GFT and exhibited a Fock coherent state with the same properties. It is clear that this correspondence will continue to $n$-particle condensates. We note that not all states in the two approaches have been studied. Since the GFT Fock space contains vertices with many open legs as one particle states, there are states and operations that can not easily be translated to $\fockSpace$. A closer inspection of the correspondence between GFT states and those in $\fockSpace$ would be desirable. 
\startNewParagraph
Lastly, we considered the physical interpretations of these Fock states by studying their geometric structure. This was done by considering candidate geometric observables associated to foliations of $\spatialSubMfd$ and their qualitative behaviour when moving through the leaves of the foliation, such as the area of the leaves and the enclosed volume. The same observables were considered both, in some classical geometries and in single- and multi-particle Fock states. 
From these considerations come some hints that the states based on a graph with several unlinked components can be understood as describing geometries with separate, disconnected components, rather than the geometry of one connected manifold.  
However, the discussion was heuristic and did not include the imposition of diffeomorphism invariance. Whether or not the results are affected by that is to be seen. 
\startNewParagraph
It would be of interest to address the open questions posed by the results we reported. For one thing, the assumptions we made could be formally checked in the semianalytic category. Perhaps it is also possible to weaken the assumptions somewhat, while still maintaining the results we have presented. Secondly, the connection to GFT should be studied more fully and more carefully. The same is true for the geometric interpretation of the multi-particle states, as this also has relevance for the corresponding states in GFT. 
Lastly, the Hamilton constraint has been completely ignored in this work. It would be interesting to understand if the Fock structure can be used for the construction of solutions. 
\startNewParagraph
Irrespective of the possible limitations and open questions, the results obtained stand to prove that one can have a symmetric Fock structure, consistent with gravity being a bosonic theory, that is naturally arising in the context of LQG at the $\hDiff$ level and which brings out a multi-particle picture in LQG. This sheds new light on the geometry and structure of the states in $\hDiff$, rigorous connection to GFT, and new tools by using the familiar language of Fock spaces often seen in QFT now in the context of LQG.


\ack
H.S. acknowledges the contribution of the COST Action CA18108.



\section*{References}
\begin{thebibliography}{}


\bibitem{Oriti:2009zz}
D.~Oriti,
``Approaches to quantum gravity: Toward a new understanding of space, time and matter,''
Cambridge University Press, 2009,
ISBN 978-0-521-86045-1, 978-0-511-51240-7

\bibitem{Ashtekar:2004eh}
A.~Ashtekar and J.~Lewandowski,
``Background independent quantum gravity: A Status report,''
Class. Quant. Grav. \textbf{21} (2004), R53
doi:10.1088/0264-9381/21/15/R01
[arXiv:gr-qc/0404018 [gr-qc]].

\bibitem{Thiemann:2007pyv}
T.~Thiemann,
``Modern Canonical Quantum General Relativity,''
Cambridge University Press, 2007,
ISBN 978-0-511-75568-2, 978-0-521-84263-1
doi:10.1017/CBO9780511755682

\bibitem{Rovelli:2011eq}
C.~Rovelli,
``Zakopane lectures on loop gravity,''
PoS \textbf{QGQGS2011} (2011), 003
doi:10.22323/1.140.0003
[arXiv:1102.3660 [gr-qc]].

\bibitem{Freidel:2005qe}
L.~Freidel,
``Group field theory: An Overview,''
Int. J. Theor. Phys. \textbf{44} (2005), 1769-1783
doi:10.1007/s10773-005-8894-1
[arXiv:hep-th/0505016 [hep-th]].

\bibitem{Oriti:2013aqa}
D.~Oriti,
``Group field theory as the 2nd quantization of Loop Quantum Gravity,''
Class. Quant. Grav. \textbf{33} (2016) no.8, 085005
doi:10.1088/0264-9381/33/8/085005
[arXiv:1310.7786 [gr-qc]].

\bibitem{Gielen:2013kla}
S.~Gielen, D.~Oriti and L.~Sindoni,
``Cosmology from Group Field Theory Formalism for Quantum Gravity,''
Phys. Rev. Lett. \textbf{111} (2013) no.3, 031301
doi:10.1103/PhysRevLett.111.031301
[arXiv:1303.3576 [gr-qc]].

\bibitem{Gielen:2013naa}
S.~Gielen, D.~Oriti and L.~Sindoni,
``Homogeneous cosmologies as group field theory condensates,''
JHEP \textbf{06} (2014), 013
doi:10.1007/JHEP06(2014)013
[arXiv:1311.1238 [gr-qc]].

\bibitem{Marchetti:2020umh}
L.~Marchetti and D.~Oriti,
``Effective relational cosmological dynamics from Quantum Gravity,''
JHEP \textbf{05} (2021), 025
doi:10.1007/JHEP05(2021)025
[arXiv:2008.02774 [gr-qc]].

\bibitem{Jercher:2021bie}
A.~F.~Jercher, D.~Oriti and A.~G.~A.~Pithis,
``Emergent cosmology from quantum gravity in the Lorentzian Barrett-Crane tensorial group field theory model,''
JCAP \textbf{01} (2022) no.01, 050
doi:10.1088/1475-7516/2022/01/050
[arXiv:2112.00091 [gr-qc]].

\bibitem{Oriti:2014yla}
D.~Oriti, J.~P.~Ryan and J.~Th\"urigen,
``Group field theories for all loop quantum gravity,''
New J. Phys. \textbf{17} (2015) no.2, 023042
doi:10.1088/1367-2630/17/2/023042
[arXiv:1409.3150 [gr-qc]].

\bibitem{Oriti:2014uga}
D.~Oriti,
``Group Field Theory and Loop Quantum Gravity,''
[arXiv:1408.7112 [gr-qc]].

\bibitem{Oriti:2017ave}
D.~Oriti,
``Group Field Theory and Loop Quantum Gravity,''
doi:10.1142/9789813220003\_0005

\bibitem{Varadarajan:1999it}
M.~Varadarajan,
``Fock representations from U(1) holonomy algebras,''
Phys. Rev. D \textbf{61} (2000), 104001
doi:10.1103/PhysRevD.61.104001
[arXiv:gr-qc/0001050 [gr-qc]].

\bibitem{Varadarajan:2001nm}
M.~Varadarajan,
``Photons from quantized electric flux representations,''
Phys. Rev. D \textbf{64} (2001), 104003
doi:10.1103/PhysRevD.64.104003
[arXiv:gr-qc/0104051 [gr-qc]].

\bibitem{Varadarajan:2002ht}
M.~Varadarajan,
``Gravitons from a loop representation of linearized gravity,''
Phys. Rev. D \textbf{66} (2002), 024017
doi:10.1103/PhysRevD.66.024017
[arXiv:gr-qc/0204067 [gr-qc]].

\bibitem{Assanioussi:2022rkf}
M.~Assanioussi and J.~Lewandowski,
``Loop representation and r-Fock measures for SU(N) gauge theories,''
Phys. Rev. D \textbf{105} (2022) no.10, 104025
doi:10.1103/PhysRevD.105.104025
[arXiv:2204.07119 [gr-qc]].

\bibitem{Lewandowski:2014hza}
J.~Lewandowski and H.~Sahlmann,
``Symmetric scalar constraint for loop quantum gravity,''
Phys. Rev. D \textbf{91}, no.4, 044022 (2015)
doi:10.1103/PhysRevD.91.044022
[arXiv:1410.5276 [gr-qc]].

\bibitem{Assanioussi:2018zit}
M.~Assanioussi,
``Polymer quantization of connection theories: Graph coherent states,''
Phys. Rev. D \textbf{98} (2018) no.4, 045016
doi:10.1103/PhysRevD.98.045016
[arXiv:1805.05299 [hep-th]].

\bibitem{Assanioussi:2020fsz}
M.~Assanioussi,
``Graph coherent states for loop quantum gravity,''
Phys. Rev. D \textbf{101} (2020) no.12, 124022
doi:10.1103/PhysRevD.101.124022
[arXiv:2004.08876 [gr-qc]].

\bibitem{Thiemann:2002vj}
T.~Thiemann,
``Complexifier coherent states for quantum general relativity,''
Class. Quant. Grav. \textbf{23}, 2063-2118 (2006)
doi:10.1088/0264-9381/23/6/013
[arXiv:gr-qc/0206037 [gr-qc]].

\bibitem{Sahlmann:2006qs}
H.~Sahlmann,
``Exploring the diffeomorphism invariant Hilbert space of a scalar field,''
Class. Quant. Grav. \textbf{24} (2007), 4601-4616
doi:10.1088/0264-9381/24/18/003
[arXiv:gr-qc/0609032 [gr-qc]].

\bibitem{Ashtekar:1995zh}
A.~Ashtekar, J.~Lewandowski, D.~Marolf, J.~Mourao and T.~Thiemann,
``Quantization of diffeomorphism invariant theories of connections with local degrees of freedom,''
J. Math. Phys. \textbf{36} (1995), 6456-6493
doi:10.1063/1.531252
[arXiv:gr-qc/9504018 [gr-qc]].

\bibitem{Ashtekar:1993wf}
A.~Ashtekar and J.~Lewandowski,
``Representation theory of analytic holonomy C* algebras,''
[arXiv:gr-qc/9311010 [gr-qc]].

\bibitem{Lewandowski:2005jk}
J.~Lewandowski, A.~Okolow, H.~Sahlmann and T.~Thiemann,
``Uniqueness of diffeomorphism invariant states on holonomy-flux algebras,''
Commun. Math. Phys. \textbf{267} (2006), 703-733
doi:10.1007/s00220-006-0100-7
[arXiv:gr-qc/0504147 [gr-qc]].

\bibitem{Fleischhack:2004jc}
C.~Fleischhack,
``Representations of the Weyl algebra in quantum geometry,''
Commun. Math. Phys. \textbf{285} (2009), 67-140
doi:10.1007/s00220-008-0593-3
[arXiv:math-ph/0407006 [math-ph]].

\bibitem{Baez:1995zx}
J.~C.~Baez and S.~Sawin,
``Functional integration on spaces of connections,''
[arXiv:q-alg/9507023 [math.QA]].

\bibitem{Lewandowski:1999qr}
J.~Lewandowski and T.~Thiemann,
``Diffeomorphism invariant quantum field theories of connections in terms of webs,''
Class. Quant. Grav. \textbf{16} (1999), 2299-2322
doi:10.1088/0264-9381/16/7/311
[arXiv:gr-qc/9901015 [gr-qc]].

\bibitem{Fleischhack:2003vk}
C.~Fleischhack,
``Proof of a conjecture by Lewandowski and Thiemann,''
Commun. Math. Phys. \textbf{249} (2004), 331-352
doi:10.1007/s00220-004-1052-4
[arXiv:math-ph/0304002 [math-ph]].

\bibitem{Ashtekar:1994mh}
A.~Ashtekar and J.~Lewandowski,
``Projective techniques and functional integration for gauge theories,''
J. Math. Phys. \textbf{36} (1995), 2170-2191
doi:10.1063/1.531037
[arXiv:gr-qc/9411046 [gr-qc]].

\bibitem{Ashtekar:1997fb}
A.~Ashtekar and J.~Lewandowski,
``Quantum theory of geometry. 2. Volume operators,''
Adv. Theor. Math. Phys. \textbf{1} (1998), 388-429
doi:10.4310/ATMP.1997.v1.n2.a8
[arXiv:gr-qc/9711031 [gr-qc]].

\bibitem{Gielen:2019kae}
S.~Gielen and A.~Polaczek,
``Generalised effective cosmology from group field theory,''
Class. Quant. Grav. \textbf{37} (2020) no.16, 165004
doi:10.1088/1361-6382/ab8f67
[arXiv:1912.06143 [gr-qc]].

\bibitem{Gielen:2016dss}
S.~Gielen and L.~Sindoni,
``Quantum Cosmology from Group Field Theory Condensates: a Review,''
SIGMA \textbf{12} (2016), 082
doi:10.3842/SIGMA.2016.082
[arXiv:1602.08104 [gr-qc]].

\bibitem{Ashtekar:1996eg}
A.~Ashtekar and J.~Lewandowski,
``Quantum theory of geometry. 1: Area operators,''
Class. Quant. Grav. \textbf{14} (1997), A55-A82
doi:10.1088/0264-9381/14/1A/006
[arXiv:gr-qc/9602046 [gr-qc]].

\bibitem{Brunnemann:2007ca}
J.~Brunnemann and D.~Rideout,
``Properties of the volume operator in loop quantum gravity. I. Results,''
Class. Quant. Grav. \textbf{25} (2008), 065001
doi:10.1088/0264-9381/25/6/065001
[arXiv:0706.0469 [gr-qc]].

\bibitem{Brown:1994py}
J.~D.~Brown and K.~V.~Kuchar,
``Dust as a standard of space and time in canonical quantum gravity,''
Phys. Rev. D \textbf{51} (1995), 5600-5629
doi:10.1103/PhysRevD.51.5600
[arXiv:gr-qc/9409001 [gr-qc]].

\bibitem{Rovelli:1993bm}
C.~Rovelli and L.~Smolin,
``The Physical Hamiltonian in nonperturbative quantum gravity,''
Phys. Rev. Lett. \textbf{72} (1994), 446-449
doi:10.1103/PhysRevLett.72.446
[arXiv:gr-qc/9308002 [gr-qc]].

\bibitem{Giesel:2012rb}
K.~Giesel and T.~Thiemann,
``Scalar Material Reference Systems and Loop Quantum Gravity,''
Class. Quant. Grav. \textbf{32} (2015), 135015
doi:10.1088/0264-9381/32/13/135015
[arXiv:1206.3807 [gr-qc]].

\bibitem{Lewandowski:2015xqa}
J.~Lewandowski and H.~Sahlmann,
``Loop quantum gravity coupled to a scalar field,''
Phys. Rev. D \textbf{93} (2016) no.2, 024042
doi:10.1103/PhysRevD.93.024042
[arXiv:1507.01149 [gr-qc]].


\bibitem{Giesel:2018opa}
K.~Giesel, A.~Herzog and P.~Singh,
``Gauge invariant variables for cosmological perturbation theory using geometrical clocks,''
Class. Quant. Grav. \textbf{35}, no.15, 155012 (2018)
doi:10.1088/1361-6382/aacda2
[arXiv:1801.09630 [gr-qc]].

\bibitem{Giesel:2020xnb}
K.~Giesel, B.~F.~Li and P.~Singh,
``Relating dust reference models to conventional systems in manifestly gauge invariant perturbation theory,''
Phys. Rev. D \textbf{104}, no.2, 023501 (2021)
doi:10.1103/PhysRevD.104.023501
[arXiv:2012.14443 [gr-qc]].

\end{thebibliography}
\end{document}